\documentclass[a4paper]{article}

	\usepackage[utf8]{inputenc}
	\usepackage[hyperref,standard,thmmarks]{ntheorem}
	\usepackage[pdftex, colorlinks=true, linkcolor = blue]{hyperref}
	\usepackage[all,cmtip]{xy}
	\usepackage{
    tikz,
    caption,
		a4,
		xspace,
		amsmath,
		amssymb,
		cleveref,
    xcolor,
	}
  \colorlet{red0}{black!75!red!100!}
  \colorlet{red1}{black!50!red!100!}
  \colorlet{red2}{black!25!red!100!}
  \colorlet{red3}{red}
  \colorlet{red4}{red!50!}
  \colorlet{red5}{red!25!}
  \colorlet{red6}{red!10!}
  \colorlet{red7}{red!5!}

	\renewtheorem{theorem}{Theorem}[section]
	\renewtheorem{proposition}[theorem]{Proposition}
	\renewtheorem{lemma}[theorem]{Lemma}
	\renewtheorem{corollary}[theorem]{Corollary}
	\theorembodyfont{\upshape}
	\renewtheorem{example}[theorem]{Example}
	\renewtheorem{remark}[theorem]{Remark}
	\renewtheorem{definition}[theorem]{Definition}

  \newcommand{\NN}{\mathbb N}
  \newcommand{\N}{\mathcal Y}
  \newcommand{\Nb}{\mathbf Y}
  \newcommand{\Nn}{Y}
  \newcommand{\ZZ}{\mathbb Z}
  
  \newcommand{\RR}{\mathbb R}

  \newcommand{\B}{\mathcal B}
  \newcommand{\M}{\mathcal Z}
  \newcommand{\MM}{\mathbf Z}
  \newcommand{\Mn}{Z}
  \newcommand{\lb}{\mathrm{lb}}
  
  \newcommand{\XX}{\mathbf{X}}
  \newcommand{\YY}{\mathbf{Y}}
  \newcommand{\mto}{\rightrightarrows}
  \newcommand{\dd}{\mathrm d}
  \newcommand{\Time}{\mathrm{time}}
  \newcommand{\length}[1]{\left|#1\right|}
  \newcommand{\flength}[1]{\left|#1\right|}
  \newcommand{\sdzero}{\textup{\texttt{0}}}
  \newcommand{\sdone}{\textup{\texttt{1}}}
  \newcommand{\albe}{\Sigma}
  \newcommand{\binalbe}{\ensuremath{\left\{\sdzero,\sdone\right\}}\xspace}
  \newcommand{\Reg}{\reg}
  \newcommand{\eq}{\operatorname{eq}}
  \newcommand{\img}{\operatorname{img}}
  \newcommand{\bigo}{\mathcal O}
  \newcommand{\demph}[1]{\textbf{#1}}
  
  \newcommand{\eval}{\mathrm{eval}}
  \newcommand{\size}[1]{\operatorname{ent}(#1)}

  \newcommand{\diam}{\operatorname{diam}}
  
  \newcommand{\LL}[1]{{\mathrm L}^{#1}}
  
  \newcommand{\Lp}{\LL{p}}
  
  \newcommand{\xip}{\xi_{\Lp}}
  \newcommand{\dom}{\mathrm{dom}}
  \newcommand{\str}{\mathbf}
  \newcommand{\reg}{\albe^{**}}

  \newcommand{\abs}[1]{\left|#1\right|}
  \newcommand{\norm}[1]{\left\|#1\right\|}

\title{Bounded time computation\\on metric spaces and Banach spaces}
\author{Matthias Schröder\footnote{Technische Universit\"at Darmstadt; Email: matthias.schroeder@cca-net.de}\ \  and Florian Steinberg\footnote{
Technische Universität Darmstadt; Email: steinberg@mathematik.tu-darmstadt.de}
}
\date{}

\begin{document}

	\maketitle
	\abstract{
		We extend the framework by Kawamura and Cook for investigating computational complexity for operators occuring in analysis.
		This model is based on second-order complexity theory for functionals on the Baire space, which is lifted to metric spaces by means of representations.
		Time is measured in terms of the length of the input encodings and the required output precision.

		We propose the notions of a complete representation and of a regular representation.
		We show that complete representations ensure that any computable function has a time bound.
		Regular representations generalize Kawamura and Cook's more restrictive notion of a second-order representation, while still guaranteeing fast computability of the length of the encodings.

		Applying these notions, we investigate the relationship between purely metric properties of a metric space and the existence of a representation such that the metric is computable within bounded time.
		We show that a bound on the running time of the metric can be straightforwardly translated into size bounds of compact subsets of the metric space.
		Conversely, for compact spaces and for Banach spaces we construct a family of admissible, complete, regular representations that allow for fast computation of the metric and provide short encodings.
		Here it is necessary to trade the time bound off against the length of encodings. 
	}
	\newpage
	\tableofcontents
	\newpage


\section{Introduction}

	Computability theory provides a mathematical framework for algorithmic considerations about discrete structures.
	The merits and applications are well-known and appreciated by both computer scientists and mathematicians.
	One of the biggest impacts the development of computability theory had on computer science is that it became possible to give mathematical proofs of non-existence of algorithms (or real-world processes, assuming the Church-Turing hypothesis) to solve certain problems.
	This has concrete implications for programming practice: The halting problem is incomputable, which means that automatic checks for correctness are impossible in general and has seeded whole fields of research like model checking.

	Complexity theory is the resource sensitive refinement of computability theory.
	For a decidable problem it investigates bounds on the number of steps or memory space a machine solving the problem needs to take or use.
	Again, in principal one of the main merits of a formal complexity theory are proofs that certain problems are inherently difficult to solve so that one does not need to continue looking for a better algorithm.
	However, while there exist hierarchy theorems, many interesting complexity classes are notoriously difficult to separate:
	The famous $\mathcal{P}$ vs.\ $\mathcal{NP}$ problem is one of the millennium problems and unsolved.
	Thus, the main focus in complexity theory has shifted to classifying problems as hard for certain classes.
	In practice, providing $\mathcal{NP}$ hardness of a problem has similar implications as proving it to not be solvable in polynomial time as a provably better algorithm would solve a problem difficult enough for many people to give up before even trying.

	In their traditional form, computability and complexity theory are only applicable to discrete structures.
	Many applications in modern computer science and numerical analysis require to model continuous structures as well.
	Engineers want to solve differential equations, mathematicians want to use computers to search for counter examples to the Riemann Hypothesis, etc.
	The most common way to model continuous structures on digital computers are floating point numbers.
	However, the use of floating point arithmetics leads to an offset between the mathematical description of an algorithm and its implementation:
	Floating point numbers do not provide many of the properties that mathematicians expect of real numbers, for instance distributivity or associativity of the operations.
	Also, the uncontrolled error propagation can lead to unreliable results of straight-forward implementations.
	Computer scientists and numerical analysts are well aware of these issues and circumvent the implications by hand using multiple precision arithmetic with correct rounding (MPFR), interval arithmetic, etc.
	For people only interested in computability issues that are willing to sacrifice a lot of efficiency, a full solution of the above problems is provided by computable analysis, Weihrauch's TTE and his theory of representations.
	\demph{Computable analysis} is a mathematically rigorous extension of computability theory to cover continuous spaces while still relying on a realistic machine model for computations.

	While computable analysis sees its roots in one of the cornerstone papers of discrete computability theory \cite{turing1936computable} and the main developments were done in the 1950's \cite{MR0086756} and 1980's \cite{MR1795407}, its time restricted counterpart \demph{real complexity theory} has until recently been stuck on the level of real functions and point-wise considerations about operators \cite{MR1137517,MR2342601}.
	Only in 2010 the framework of \demph{second-order representations} was introduced by Kawamura and Cook \cite{Kawamura:2012:CTO:2189778.2189780} and made it possible to consider uniform complexity of operators on spaces like the continuous functions on the unit interval.
	The framework heavily relies on second-order complexity theory for functions on the Baire space which was introduced by Mehlhorn \cite{MR0411947} and in particular on a fairly recent characterization of Mehlhorns basic feasible functionals by resource bounded oracle Turing machines by Kapron and Cook \cite{MR1374053}.
	Among the early achievements of the framework of second-order representations are the characterization of the minimal representation of the continuous functions on the unit interval that renders evaluation polynomial-time computable and a number of translations to the new setting of some earlier point-wise results about complexity theory that related certain operators on the continuous functions to important inclusions of discrete complexity classes.
	Prior to Kawamura and Cook, the first author introduced a slightly different, more restrictive framework for complexity \cite{MR2090390} and provided some applications \cite{KUNKLE2005111}.
	Many of the results in this paper were found when attempting to unify these two approaches.
	
	Computable analysis seems satisfactory in the sense that it is applicable to a big variety of problems that come up in practice:
	For separable metric spaces with a canonical sequence or more generally countably based $T_0$ spaces with a canonical basis there exists a standard representation.
	It is possible to form products and disjoint unions of spaces, and there is a function space construction.
	Many results like the computable Weierstraß Approximation Theorem provide compatibility of the above constructions:
	Under reasonable assumptions on the spaces involved it does not matter whether a canonical sequence is chosen in the function space or the function space construction is used.
	The notion of admissibility provides a reasonable condition for whether or not a chosen computable structure is compatible with a topology on the space.
	While this still leaves some gaps on whose closure the community of computable analysis is actively working on \cite{MR2146883,mir_mods_00000784}, in the past years the focus has shifted to finding applications to problems that people in numerical analysis are interested in \cite{MR1694433,MR2275411,MR2275415,MR2351947,NavierStokes}.
	
	In comparison to computable analysis, the ground stock of problems that real complexity theory is applicable to is vanishingly small:
	The standard representation of metric spaces is known to only lead to a good complexity theory in very restricted cases \cite{MR1952428,Kawamura:2016:CTC:2933575.2935311}.
	In particular it does not lead to the right result for the continuous functions on the unit interval.
	The construction for the representation of this space, in turn, can not be expected to generalize to a big class of spaces \cite{MR3219039}.
	This situation is not acceptable and many efforts have been done to improve it.
	The most notable advances are the constructions of representations for the analytic functions on the unit disc and the levels of Gevrey's hierarchy \cite{MR3377508} based on prior work by Norbert Müller \cite{Muller1995}, the investigations of complexity of functionals \cite{MR3239272,postive} and advances on $L^p$ and Sobolev spaces \cite{arXiv:1612.06419,Kawamura:2016:CTC:2933575.2935311}.
	However, all of the above examples do only cover very restricted areas and handle examples in a highly specialized and not generalizable way.

	This paper aims to close the gap in applicability between computable analysis and real complexity theory where it is possible and to point out where this is impossible.
	It provides an in-depth investigation of the general restrictions of the framework of Kawamura and Cook when computations on metric spaces are carried out.
	It generalizes the Cauchy representation in a way such that the constructions for continuous functions on the unit interval and $L^p$ and Sobolev spaces can be seen as special cases.
	The construction is highly inspired by results from approximation theory and close to ideas frequently used in numerical analysis.
	We firmly believe that this provides a general recipe for constructing useful representations that are complexity theoretically well-behaved and opens the field for applications similar to those computable analysis has.
	For accessibility of the constructions the requirement of being second-order representation is relaxed to a weaker notion we call regularity.
	We believe this to be an important step as it removes the necessity of padding that seems to be an unnatural thing to do in real world applications.
	At the same time it maintains many of the advantages of second-order representations, like accessibility of the length function.
	Furthermore, we introduce the properties of completeness and computable completeness that guarantee that any computable operator allows a time bound.
	Computable completeness lifts computable admissibility to the bounded-time realm in that assuming it for a representation of a computable metric space implies bounded time equivalence to the Cauchy representation, where without it only computable equivalence is guaranteed.
	We refrain from making this a definition as other authors have claimed that a resource restricted notion of admissibility should not exists \cite{MR3219039}.

	\subsection{The content of this paper}

		The first chapter contains the introduction, this description of the content of the paper and a short section listing the most basic notational conventions.

		\Cref{sec:representations} gives a brief introduction into computable analysis formulated not in the standard way, but in an appropriate form such that second-order complexity theory is applicable.
		It also presents those results and definitions from second-order complexity theory that are of relevance for the content of the paper.
		The last part introduces the notion of a length of a representation which is new and important for the next chapter.
		It proves that many representations have a length.

		In \Cref{sec:metric entropy} the concept of metric entropy is presented as a way of measuring the size of a compact set.
		This concept is not new but has been used in approximation theory, constructive analysis and proof mining and computable analysis before.
		However, we skew the definitions a little in a way that make them fit in with the framework presented in \Cref{sec:representations}.
		The most basic properties are listed for the convenience of the reader.
		Then the first major result of this paper is stated:
		\Cref{resu:from complexity to metric entropy} constructs bounds on the metric entropy of certain compact subsets from a running time of the metric.
		Let us state the result of the theorem, when applied to compact spaces in natural language:
		Whenever there exists an encoding of a compact metric space $\Mn$ by string functions such that the metric can be obtained from the encodings in time $T$ and each element has an encoding of length at most $\ell$, then the size of the compact space is bounded polynomially in $T$ and $\ell$.
		More explicitly $\size\Mn\in \bigo(T(\ell, \cdot+2)^2)$.
		This is a general restriction of the framework.
		If a compact set is too big, we can not expect to find representations that both feature short names and fast computation of the metric.
		A trade-off is necessary.

		\Cref{sec:regularity and completeness} prepares the construction of interesting representations by introducing some necessary restrictions on running times.
		We introduce a concept of time-constructibility for second-order running times that restricts to the usual notion of time-constructibility for usual running times and choose a notion of monotonicity.
		We then present the notion of regularity that relaxes being a second-order representation.
		This is mainly done to make the definitions of the standard representations in the next section more legible, however, we think that it introduces an important concept that makes the framework more accessible and brings it closer to actual implementations.
		Finally, we introduce completeness and computable completeness of representations.
		The notion is inspired by the notion of properness \cite{MR2090390} and we prove it to imply that any computable function is computable in bounded time in \Cref{resu:completeness as properness}.
		\Cref{resu:completenesses} proves that a separable metric space is complete if and only if every Cauchy representation is complete and that the Cauchy representation of a complete computable metric space is computably complete.
		This is the motivation for the name and shows that if a representation of a complete computable metric space is computably complete, it is bounded-time equivalent to the Cauchy representation.

		The second to last \Cref{sec:construction of standard representations} constructs standard representations for metric spaces.
		As can be expected from the results of \Cref{sec:metric entropy} it is necessary to impose conditions on the relation between the length of names and the running time.
		We specify sets of conditions under which it is possible to construct such representations.
		However, in general it can not be guaranteed that the bounds from \Cref{resu:from complexity to metric entropy} are tight.
		To be able to provide lower bounds on the metric entropy we have to impose further conditions on the metric space.
		The two conditions under which we succeeded are those of an infinite compact metric space (\Cref{resu:representing compact spaces}) and of an infinite dimensional Banach space that allows a Schauder basis (\Cref{resu:representing banach spaces}).
		Of course in the former case the bound is only valid if it is not bigger than the metric entropy of the compact space itself which makes the statement slightly more complicated.
		The merit of the compact case is that the existence of a uniformly dense sequence is automatic.
		This is in contrast to the situation for Banach spaces, where the existence of a Schauder basis has to be assumed.
		The last part mentions that additional assumptions have to be made to guarantee that the construction works well with a computable structure on the space.
		For Banach spaces it is known that this is for principal reasons, for metric spaces it is open whether these additional assumptions are indeed necessary.

		Finally, \Cref{sec:appl} provides applications of the constructions and connections to other results.
		First it compares the results to the results from approximation theory that inspired the constructions for Banach spaces in the first place.
		These results in turn can be applied to the construction and improve some of the bounds drastically.
		This did not come as an surprise to the authors:
		While we consider the proofs within the framework of real complexity theory elegant, they are new and apply to a more general situation.
		The authors from approximation theory, in contrast, have had a lot of time and invested a lot of effort into optimizations.
		The last two parts of \Cref{sec:appl} go on to reconstruct different known representation by choosing the parameters in the constructions from \Cref{resu:representing compact spaces,resu:representing banach spaces} appropriately.
		Most notably the representation of the continuous functions on the unit interval that has been introduced and proven to be the minimal second-order representation to allow polynomial-time evaluation by Kawamura and Cook.

		The paper ends in a conclusion and the bibliography.

	\subsection{Notational conventions}

    By $\Sigma$ we denote the alphabet $\{\sdzero,\sdone\}$
    and by $\albe^*$ 
    we denote the set of finite strings over $\Sigma$.
    The letters $\str a,\str b,\ldots$ will be used for strings.
    By the Baire space
    we mean the space $\B := (\albe^*)^{\albe^*}$ of total functions $\varphi,\psi,\ldots$ from finite strings to finite strings equipped with the product topology.
    It is topologically equivalent to the countable product of the discrete natural numbers.
    We assume familiarity with the notions of computability and complexity of total and partial string functions.

    We identify the set $\NN$ of natural numbers with the set
    $\{\varepsilon,\sdone,\sdone\sdzero,\sdone\sdone,\ldots\}$
    of their binary representations.
    Another copy of the natural numbers is the set $\omega := \{\varepsilon,\sdone,\sdone\sdone,\sdone\sdone\sdone,\ldots\}$ of their unary representations.
    We denote elements of both $\NN$ and $\omega$ by $n,m,\ldots$.
    We identify the set of integers with the set $\ZZ:=\NN\setminus\{\varepsilon\}\cup\sdzero\NN$, where $\sdzero n$ is interpreted as $-n$ (note that the union is disjoint).

    By $|\cdot|\colon \albe^*\to \omega$ the length mapping is denoted, which simply replaces all occurrences of $\sdzero$ by $\sdone$s.
    Computations on products are handled via \demph{tuple functions}.
		\Cref{sec:sub:metric entropy and complexity} requires the tuple functions to have some very specific properties.
		Up until then any standard tuple function (i.e.\ some bijective, polynomial-time computable function with polynomial-time computable projections) may be used.
		For any given dimension $k$ define an injection $(\albe^*)^k \to \albe^*$ as follows:
		Given strings $\str a_1,\ldots, \str a_k$ denote by $\str c_i= c_{i,1}\ldots c_{i,\max\{\length{\str a_i}\}+1}$ the padding of $\str a_i$ to length $\max\{\length{\str a_i}\mid i=1,\ldots,k\}+1$ by appending a $\sdone$ and then an appropriate number of $\sdzero$s.
		Set
		\[ \langle \str a_1,\ldots,\str a_k\rangle := c_{1,1}\ldots c_{k,1}c_{1,2}\ldots c_{k,2}\ldots\ldots c_{1,\max\{\length{\str a_i}\}+1}\ldots c_{k,\max\{\length{\str a_i}\}+1}. \]
		These tuple functions are not surjective thus extra care has to be taken in some definitions.
		However, for a fixed $k$ each of the tuple functions and its projections
		\begin{equation}\tag{$\pi$}\label{eq:projection}
			\pi_i(\str b) := \begin{cases} \sdzero a_i& \text{ if } \str b = \langle\str a_1,\dots,\str a_d\rangle \\ \varepsilon &\text{ if }\str b \notin \img(\langle\cdot,\ldots,\cdot\rangle) \end{cases}
		\end{equation}
		are computable in linear time.
		The encoding of finite sequences of arbitrary length as $\langle k,\langle \str a_1,\ldots,\str a_k\rangle\rangle$ is explicitly written out when used to avoid confusion.
		The corresponding \demph{pairing function for string functions} is defined as follows:
		\[ \langle\cdot,\cdot\rangle:\Sigma^{*\Sigma^*}\times\Sigma^{*\Sigma^*},\quad \langle \varphi,\psi\rangle(\str a) := \langle\varphi(\str a),\psi(\str a)\rangle. \]
		The paring function for string functions is injective (and bijective if a bijective pairing function is used).

    A multivalued function $f:\subseteq X\mto Y$ is a function from $X$ to the power set of $Y$.
		The elements of $f(x)$ are interpreted as the \lq acceptable return values\rq.
		The domain $\dom(f)$ of $f$ is the set of $x$ such that $f(x)\neq \emptyset$.
		A multivalued function is total if $\dom(f) = X$ and this is indicated by writing $f:X\mto Y$.
		We choose the convention that a function is assumed to be total if it is not explicitly stated that it is allowed to be partial. 
		A partial function $f:\subseteq X\to Y$ is identified with the multivalued function $\tilde f:\subseteq X\mto Y$ with $\tilde f(x)= \{f(x)\}$ if $x\in\dom(f)$ and $\tilde f(x)=\emptyset$ otherwise.

    The reader is assumed to be familiar with the basic theory of metric spaces $(\Nn,d)$ such as the diameter of a set and its distance function.
    The closed ball of radius $r\geq 0$ around $x\in \Nn$ is denoted by
    $ B^c_r(x) := \{y\in \Nn\mid d(x,y) \leq r\} $
    and the open ball by
    $B^o_r(x):=\{y\in  \Nn\mid d(x,y)<r\}$.
    The open balls form a base of the topology induced by the metric.
    Note that $B^o_r(x) \subseteq \overline{B^o_r(x)}\subseteq B^c_r(x)$, but in general nothing further can be said about strictness of the inclusions.
    In particular, the closed balls may in general be different from the closures of the open balls.

    All vector spaces considered in this paper have the space $\RR$ of the real numbers as their underlying field. 

\section{Representations}\label{sec:representations}

	Discrete computability theory encodes objects by finite strings to make them accessible to digital computers.
	This only works if the structure considered is countable.
	Most metric spaces appearing in practice, however, are not countable.
	For instance the real numbers, or, to mention a compact one, the unit interval are uncountable.
	Computable analysis removes the necessity of countability by encoding elements by infinite objects (infinite binary strings or string functions) instead of strings.

	For us the \demph{Baire space} is the space of all string functions $\B:=(\albe^*)^{\albe^*}$ equipped with the product topology.

	\begin{definition}
		A \demph{represented space} is a pair $\XX=(X,\xi)$ of a set $X$ and a partial surjective mapping $\xi:\subseteq\B\to X$.
		The elements of $\xi^{-1}(x)$ are called the \demph{names} (or \demph{$\xi$-names} or \demph{$\XX$-names}) of $x$.
	\end{definition}

	A space with a fixed representation is called a \demph{represented space}.
	We denote represented spaces by $\XX$, $\YY,\ldots$, their underlying sets by $X$, $Y,\ldots$ and their representations by $\xi_{\XX}$, $\xi_{\YY},\ldots$
	Like the topology of a topological space the representation of a represented space is only mentioned explicitly if necessary to avoid ambiguities.
	An element of a represented space is called \demph{computable} if it has a computable name.
	It is said to lie within a complexity class if it has a name from that complexity class.

	On one hand, any represented space carries a natural topology: The final topology of the representation.
	On the other hand, one often looks for a representation suitable for a topological space.
	It is reasonable to require such a representation to induce the topology the space is equipped with.
	For this, continuity is necessary but not sufficient.
	Continuity together with openness is sufficient but not necessary.

	The following is the standard representation for carrying out computability considerations on metric spaces \cite{MR1795407} and has also been used in different contexts for a long time \cite{MR966421,MR1232903}.
	\begin{definition}\label{def:Cauchy representation}
		Let $\M:=(\Mn,d,(r_i))$ be a triple such that $(\Mn,d)$ is a separable metric space and $(r_i)_{i\in\NN}$ is a dense sequence in $\Mn$.
		Define the \demph{Cauchy representation} $\xi_{\M}$ of $\Mn$ with respect to $(r_i)$ as follows:
		A string function $\varphi\in\B$ is a $\xi_{\M}$-name of $x\in M$ if and only if for all $n\in\NN$
		\[ \varphi(n)\in\NN\wedge d(x,r_{\varphi(n)}) \leq \frac1{n+1}. \]
	\end{definition}
	This representation can be checked to induce the metric topology.
	Note that equivalently one could use integers in unary and require $d(x,r_{\varphi(\sdone^n)})\leq 2^{-n}$, which is the more common convention in real complexity theory.
	A third equivalent option would have been to use encodings of rational numbers.

	The following representation is the standard representation of real numbers that is used throughout real complexity theory:
	\begin{definition}\label{def:standard representation of the reals}
	Let $\varphi\in\B$ be an $\RR$-name of $x\in\RR$ if and only if for all $n\in\NN$
	\[ \varphi(n)\in\ZZ\wedge\abs{x-\frac{\varphi(n)}{n+1}}\leq \frac1{n+1}. \]
	\end{definition}
	This representation is in a very strong sense (see \Cref{def:equivalence}) equivalent to the Cauchy representation, if the dense sequence is chosen as a standard enumeration of the dyadic numbers.

	Recall the pairing $\langle\cdot,\cdot\rangle:\B\times\B\to\B$ on string functions from the introduction.
	\begin{definition}\label{def:product representation}
		Let $\XX$ and $\YY$ be represented spaces.
		Define their \demph{product} $\XX\times \YY$ by equipping $X\times Y$ with the representation $\xi_{\XX\times\YY}$ defined by
		\begin{equation*} \xi_{\XX\times \YY}(\langle\varphi,\psi\rangle) = (x,y)\\\Leftrightarrow \\\big(x=\xi_\XX(\varphi) \wedge y=\xi_\YY(\psi)\big).
		\end{equation*}
	\end{definition}
	That is: a string function is a name of the pair $(x,y)$ if and only if it is the pairing of a name of $x$ and a name of $y$.

	Computing functions between represented spaces can be reduced to computing functionals on the Baire space:
	\begin{definition}\label{def:realizer}
		Let $f:\XX\mto\YY$ be a multivalued function between represented spaces.
		A partial function $F:\subseteq\B\to\B$ is a \demph{realizer} of $f$, if
	\[ \varphi\in\xi_\XX^{-1}(x) \quad \Rightarrow \quad F(\varphi)\in \xi_\YY^{-1}(f(x)). \]
	\end{definition}
	That is: The realizer $F$ translates each name of $x$ into some name of $f(x)$.
	No assumptions about the behavior on elements that are not names are made.

	Computability of functionals on the Baire space was introduced in \cite{MR0051790} and is well established.
	An overview over the historical development can be found in \cite{MR2143877}.
	We use the model of computation by oracle Turing machines.
	A functional $F:\subseteq\B\to\B$ is called \demph{computable} if there is an oracle Turing machine $M^?$ such that the computation of $M^?$ with oracle $\varphi$ and on input $\str a$ halts with the string $F(\varphi)(\str a)$ written on the output tape, or for short if $M^\varphi(\str a) = F(\varphi)(\str a)$.
	A function between represented spaces is called \demph{computable} if it has a computable realizer.

	\begin{example}[Comp. metric spaces]\label{ex:computable metric space}
		A triple $\M=(\Mn,d,(r_i))$ is a \demph{computable metric space}, if $(\Mn,d)$ is a separable metric space and $(r_i)$ is a dense sequence such that the discrete metric
		\begin{equation}\tag{$\tilde d$}\label{eq:the discrete metric}
			\tilde d: \NN\times \NN \to \RR, \quad (i,j)\mapsto d(r_i,r_j)
		\end{equation}
		is computable.
		To be specific here, a representation of $\NN$ must be picked:
		Set $\xi_{\NN}(\varphi)=n$ if and only if $\varphi(\str a) = n$ for all $\str a\in\Sigma^*$.
		This is the Cauchy representation with respect to the discrete metric and $(n)_{n\in\NN}$ as dense sequence.

		The metric $d:\Mn\times \Mn\to \RR$ of a computable metric space is computable with respect to the Cauchy representation $\xi_{\M}$ from \Cref{def:Cauchy representation} and the representation from \Cref{def:standard representation of the reals} on $\RR$:
		To specify an oracle Turing machine computing a name of $d(x,y)$ from a pair $\langle\varphi,\psi\rangle$ of names of $x$ and $y$ as oracle, on input $n$ do the computation of $\tilde d$ on inputs $i:=\varphi(4n+3)$, $j:=\psi(4n+3)$ and input $2n+1$.
	\end{example}

	\subsection{Second-order complexity theory}

		Complexity theory for functionals on the Baire space is called \demph{second-order complexity theory}.
		It was originally introduced by Mehlhorn \cite{MR0411947}.
		This paper uses a characterization via resource bounded oracle Turing machines due to Kapron and Cook \cite{MR1374053} as definition.
		Such a machine is granted time depending on the size of the input.
		The string functions are considered part of the input.

		\begin{definition}
			For $\varphi\in\B$ a string function define its \demph{length} $\flength\varphi:\omega\to \omega$ of $\varphi$ by
			\[ \flength{\varphi}(n) := \max\{\length{\varphi(\str a)}\mid \length{\str a} \leq n\}. \]
		\end{definition}
		For instance: The length of each polynomial-time computable string function is bounded by an polynomial from above.
		Polynomials $\omega\to \omega$ have themselves as length.
		A polynomial as function $\NN\to\NN$ has linear length where the slope is the degree of the polynomial.

		A running time is a mapping that assigns an allowed number of steps to sizes of the inputs.
		Thus, it should be of type $\omega^\omega\times \omega \to \omega$.
		The following conventions are used to measure the time of oracle interactions of an oracle Turing machine:
		Writing the query takes time.
		Reading the answer tape takes time.
		Writing the answer to an oracle query to the answer tape is done by the oracle and does not take time.
		It may very well happen that some of the content of the answer tape is not accessible to the machine due to running time restrictions.
		\begin{definition}\label{def:M runs in time T}
			An oracle Turing machine $M^?$ \demph{runs in time $T:\omega^\omega\times \omega\to\omega$}, if for each oracle $\varphi\in\B$ and input string $\str a$ it terminates within at most $T(\flength\varphi,\length{\str a})$ steps.
			It \demph{runs in time $\bigo(T)$} if there is a $C\in \omega$ such that it runs in time $(l,n)\mapsto CT(l,n) + C$.
		\end{definition}

		This definition, for instance also used in \cite{higherorder}, is a straight-forward generalization of a characterization from \cite{MR1374053} and its application to computable analysis from \cite{Kawamura:2012:CTO:2189778.2189780}.
		The latter sources only use the definition for a subclass of running times that are considered polynomial and called \demph{second-order polynomials}.
		The set of second-order polynomials is recursively defined as follows:
		\begin{itemize}
			\item Whenever $p$ is a positive integer polynomial, then the mapping $(l,n)\mapsto p(n)$ is a second-order polynomial.
			\item If $P$ and $Q$ are second-order polynomials then $P+Q$ and $P\cdot Q$ are second-order polynomials.
			\item If $P$ is a second-order polynomial then $(l,n)\mapsto l(P(l,n))$ is a second-order polynomial.
		\end{itemize}
		Second-order polynomials also turn up in other contexts \cite{MR1462200}.
		
		A functional $F:\subseteq \B\to\B$ is called \demph{computable in time $T$} resp.\ \demph{computable in polynomial time} if there is an oracle machine that computes it and runs in time $T$ resp.\ in time bounded by some second-order polynomial.

		\begin{definition}
			A function between represented spaces is \demph{polynomial-time computable} resp.\ \demph{computable in time $T$}, if it has a polynomial-time computable realizer (in the sense of \Cref{def:realizer}) resp.\ a realizer computable in time $T$.
		\end{definition}
		A function is called \demph{computable in bounded time}, if there is any running time $T$ such that it is computable in time $T$.
		While a total functional on the Baire space is always computable in bounded time, not every computable function between represented spaces is computable in bounded time as all continuous realizers may be partial.

		\begin{example}[\ref{ex:computable metric space} continued]\label{ex:computable metric spaces continued}
			Let $\M=(\Mn,d,(r_i))$ be a computable metric space.
			Let $t:\omega \times \omega \to \omega$ be a running time of the discrete metric $\tilde d(i,j):=d(r_i,r_j)$ (I.e. the function $S:\omega^\omega\times \omega\to\omega$ defined by $S(l,n):=t(l(0),n)$ is a running time, this is not a restriction).
			The algorithm to compute the metric from \Cref{ex:computable metric space} terminates within $\bigo(T)$ steps for the function $T(l,n):=t(l(n+3),n+1)$.
			In particular $d$ is polynomial-time computable if $\tilde d$ is polynomial-time computable.
		\end{example}

		Let $\xi$ and $\xi'$ be representations of the same space $X$.

		A continuous functional $F:\subseteq\B\to\B$ is called a \demph{translation} from $\xi$ to $\xi'$ if it is a realizer of the identity function, i.e. if it maps $\xi$-names into $\xi'$-names of the same element.
		\begin{definition}\label{def:equivalence}
			Two representations are called \demph{topologically equivalent} if there exist translations in both directions.
			They are called \demph{computably, resp.\ bounded time or polynomial-time equivalent} if there exist computable, resp.\ bounded time or polynomial-time computable translations in both directions.
		\end{definition}
		The Cauchy representation from \Cref{def:Cauchy representation} is polynomial-time equivalent to the open representation that arises if the inequality in its definition is made strict.
		The comment after \Cref{def:standard representation of the reals} points to polynomial-time equivalence when it says \lq equivalent\rq.
		In the following we identify representations if they are polynomial-time equivalent.

		\Cref{def:M runs in time T} implies that any machine that runs in bounded time computes a total function.
		The parts of our results that assume bounded time computability remain valid under the weaker assumption that the machine runs in bounded time, whenever the oracle is from the domain of the computed function.
		Thus, we additionally use the following weaker notion:
		\begin{definition}
			We say that an oracle Turing machine $M^?$ \demph{runs in time $T$ on a subset $A\subseteq \B$}, if its run with oracle $\varphi\in A$ on input string $\str a$ terminates within $T(\flength\varphi,\length{\str a})$ steps.
		\end{definition}
		A functional on the Baire space is \demph{computable in time $T$ on $A$} if there is a machine that computes it and runs in time $T$ on $A$.
		For functions between represented spaces the set $A$ is usually the domain of the representation of the domain of the function.
		\begin{definition}
			A multivalued function $f:\XX\mto\YY$ between represented spaces is called \demph{computable in time $T$ on $\XX$} if it has a realizer (in the sense of \Cref{def:realizer}) that is computable in time $T$ on $\dom(\xi_{\XX})$.
		\end{definition}
		
		\begin{remark}
			Kawamura and Cook's framework of second-order representations \cite{Kawamura:2012:CTO:2189778.2189780} (cf. also \Cref{sec:regularity}) can be seen to use this definition for its definition of polynomial-time computability.
			In this case, however, it turns out to be equivalent to the classical definition due to the existence of a polynomial-time retraction. 
			Note, however, that some care has to be taken, as in other cases there are subtle differences between computing on the Baire space and computing on a subset even if there is such a retraction (cf. \Cref{def:time-constructible} and the example afterwards).
		\end{remark}

	\subsection{Admissibility, length and compactness}
		Recall that a set is called \demph{relatively compact} if its closure is compact.
		The relatively compact subsets of the Baire space are easily classified.
		A subset of the Baire space is relatively compact if and only if it is contained in a set from the following family:
		\begin{definition}
			Define a family $(K_l)_{l\in\omega^\omega}$ of compact subsets of $\B$ by
			\[ K_l := \{\varphi\in \B \mid \flength{\varphi}\leq l\}. \]
		\end{definition}
		The mentioned property of this family of sets resembles hemi-compactness:
		Whenever $K$ is a compact subset of the Baire space, there is some $l\in\omega^\omega$ such that $K\subseteq K_l$.
		The only difference is that the index set $\omega^\omega$ is not countable but has the cardinality of the continuum.

		\begin{definition}
			We call a function $\ell\in\omega^\omega$ a \demph{length} of a representation $\xi$ of a space $X$, if $\xi(K_\ell)= X$.
		\end{definition}
		I.e.\ $\ell$ is a length of $\xi$ if every element has a $\xi$-name of length $\ell$.
		This condition does not imply that the domain of $\xi$ is included in $K_\ell$, but is strictly weaker.
		
		\begin{example}\label{ex:length of a representation}
			Let $\Omega$ be a bounded subset of $\RR^d$.
			And let $C\in \omega$ be such that each element of $\Omega$ has supremum norm less than $2^C$.
			Let $\xi_{\Omega}$ denote the range restriction of the $d$-fold product of the standard representation from \Cref{def:standard representation of the reals}.
			Then the function $\ell(n) := 2d(n + C + 4)$ is a length of $\xi_{\Omega}$.
		\end{example}

		Not every representation has a length, but representations of compact spaces usually do for the following reason:

		\begin{theorem}[Compactness and length]\label{resu:length of open representations}
			Any open representation of a compact space has a length.
		\end{theorem}

		\begin{proof}
			Let $\xi$ be an open representation of a compact space $K$.
			Let $\str a_m$ denote the standard enumeration $\NN\to\Sigma^*$ that simply removes the first digit.
			Recursively define a function $\ell:\omega\to\omega$ as follows:

			For $m\in \NN$ set
			\[ U_{0,m}:= \{\varphi\mid \varphi(\varepsilon)= \str a_m\}\subseteq\Sigma^{*\Sigma^*}. \]
			The sequence $(U_{0,m})_{m\in\NN}$ forms an open cover of the Baire space.
			Since $\xi$ is open, the images of these sets compose an open cover of $K$.
			Since $K$ is compact, there exists a finite subset $I$ of $\NN$ such that $(\xi(U_{0,n}))_{n\in I}$ covers $K$.
			Define
			\[ \ell(0) := \max\{\length{n} \mid n\in I\} \]
			Note that this means that each element of $K$ has a name $\varphi$ such that $\flength{\varphi}(0)=\length{\varphi(\varepsilon)}\leq \ell(0)$.

			Now assume that $\ell$ has been defined up to value $n$ such that each element of $K$ has a name $\varphi$ such that for each $k\leq n$ it holds that $\flength{\varphi}(k)\leq \ell(k)$.
			To construct $\ell(n+1)$ let $J$ be the set of all strings of length exactly $n+1$.
			Define a sequence $(U_{n+1,m})_{m\in \NN^{J}}$ as follows:
			\[ U_{n+1,(m_{\str b})_{\str b\in J}}:= \{\varphi\mid \forall k\leq n:\flength{\varphi}(k)\leq \ell(k)\text{ and } \forall \str b\in J:\varphi(\str b) = m_{\str b}\}. \]
			Note that each of these sets is an open subset of Baire space and that the assumptions of the recursion make sure that the union of their images is $K$.
			Thus, again from the compactness of $K$ it follows that there is a finite subset $I$ of $\NN^J$ such that $(\xi(U_{n+1,m}))_{m\in I}$ covers $K$. 
			Set
			\[ \ell(n+1):= \max\{\length{m}\mid \exists (n_{\str a})_{\str a\in J}\in I,\exists\str b\in J: m = n_{\str b}\}. \]
			By choice of $I$ it is again the case that each element of $K$ has a name $\varphi$ such that for each $k\leq n+1$ it holds that $\flength{\varphi}(k)\leq \ell(k)$.

			The constructed function is a length of the representation by its definition.
		\end{proof}

		All representations considered so far were polynomial-time equivalent to open representations.
	This is due to the fact that only representations of metric spaces were considered and metric spaces are second countable.
	The abstract concept from computable analysis that remains appropriate in the general case is the following:
	\begin{definition}
		A representation $\xi$ of a topological space $X$ is called \demph{admissible}, if for any continuous representation $\xi'$ there is a translation (that is: A continuous realizer of the identity) to $\xi$.
	\end{definition}
	A representation of a second countable Hausdorff space is admissible if and only if it is topologically equivalent to an open representation \cite{MR1923900}.
	
	The following corollary can be obtained from \Cref{resu:length of open representations} together with the observation, that any admissibly represented compact Hausdorff space is a metrizable space (for instance \cite[Proposition 3.3.2]{SchroederPhD}).
	\begin{corollary}[Length of restrictions]\label{resu:length of restrictions}
		Let $\xi$ be an admissible representation of a Hausdorff space $\XX$.
		Then the range restriction of $\xi$ to any compact subset $K\subset \XX$ has a length.
	\end{corollary}

\section{Metric entropy}\label{sec:metric entropy}

	Throughout this chapter let $\N$ be a represented metric space, that is $\N=(\Nb,d)$ is a pair of a represented space $\Nb = (\Nn,\xi)$ and a metric $d$.
	We do not impose any compatibility conditions, however, most results assume the metric to be computable in bounded time, this in turn implies continuity of the metric with respect to the topology that $\Nb$ has as a represented space.

	It is well known that in a complete metric space a subset is relatively compact if and only if it is totally bounded.
	The following notion is a straight-forward quantification of total boundedness and can be used to measure the \lq size\rq\ of compact subsets of metric spaces.
	It was first considered in \cite{MR0112032}, where some of the names originate from.
	It is also regularly used in proof mining \cite{MR2445721,Kohlenbach} and constructive mathematics \cite{MR0221878}, where other names are taken from.
	A comprehensive overview can be found in \cite{lorentz1966}.

	\begin{definition}\label{def:metric entropy and size}
		A function $\nu:\omega \to \omega$ is called \demph{modulus of total boundedness} of a subset $K$ of a metric space, if for any $n\in\omega$ the set $K$ can be covered by $2^{\nu(n)}$ closed balls of radius $2^{-n}$.
		The smallest modulus of total boundedness of a set $K$ is called the \demph{metric entropy} of the set and denoted by $\size K$.
	\end{definition}

	Due to the use of closed balls it holds that $\size K=\size{\overline{K}}$.
	The metric entropy of a compact set is often hard to get hold of while upper bounds are easily specified.
	In a complete metric space a closed set allows for a modulus of total boundedness if and only if it is compact.

	A modulus of total boundedness should intuitively be understood as an upper bound on the size of a compact set.
	For providing lower bounds, a different notion is more convenient.

	\begin{definition}\label{def:spanning bound}
		A function $\eta:\omega\to\omega$ is called a \demph{spanning bound} of a subset $K\subseteq \Nn$ if for any $n$ there exist elements $x_1,\ldots,x_{2^{\eta(n)}}$ such that for all $i,j\in\NN$
		\[ i\neq j \quad \Rightarrow\quad d(x_i,x_j)> 2^{-n+1}. \]
	\end{definition}
	This means, that the closed $2^{-n}$-balls around the points $x_i$ are disjoint.
	It is not difficult to see that there is a biggest spanning bound if and only if the set $K$ is relatively compact.
	The following states that spanning bounds are lower bounds to the metric entropy:

	\begin{proposition}\label{resu:spanning bounds}
		Let $K$ be a subset of a metric space.
		For any spanning bound $\eta$ of $K$ it holds that $\eta(n)\leq\size K(n)$.
	\end{proposition}

	\begin{proof}
		Show by contradiction that $\size K(n)+1$ is not a value of a spanning bound.
		Assume it is a value of a spanning bound.
		This means that there are $2^{\nu(n)+1}$ elements $x_i$ such that $d(x_i,x_j)\geq 2^{-n+1}$.
		By definition of $\size K$ there are $2^{\nu(n)}$ elements $y_i$ such that the closed $2^{-n}$-balls around the $y_i$ cover $K$.
		Since each $x_i$ has to lie in one of the balls and there are more $x_i$ than $y_i$, there are indices $j$, $l$ and $m$ such that $x_j$ and $x_l$ both lie in the closed $2^{-n}$-ball around $y_m$.
		By triangle inequality it follows that
		\[ d(x_j,x_l) \leq d(x_j,y_m) + d(y_m,x_l) < 2 \cdot 2^{-n}, \]
		which is a contradiction to $\size K(n)+1$ being a value of a spanning bound.
	\end{proof}

	\subsection{Metric entropy in normed spaces}
		In general, the metric entropy heavily depends on the metric chosen for a space, and not only on the topology it induces.
		For normed spaces, however, the situation is a lot better.
		Recall that a norm $\norm\cdot$ on a vector space $X$ induces a metric by $d(x,y):= \norm{x-y}$.
		But not every metric vector space can be equipped with a norm.

		\begin{proposition}\label{resu:equivalent norms}
			Let $\|\cdot\|$ and $\|\cdot\|'$ be norms on a vector space $X$ that induce the same topology.
			Then there exists a constant $C\in\omega$ such that whenever $\nu$ is a modulus of total boundedness of a subset $K$ of $X$ with respect to $\|\cdot\|$, then $\nu'(n):=\nu(n+C)$ is a modulus of total boundedness of $K$ with respect to $\|\cdot\|'$ and vice versa.
		\end{proposition}

		\begin{proof}
			Note, that norms inducing the same topology are equivalent.
			This means that there exists a constant $C$ such that for any $x\in X$
			\[ 2^{-C}\|x\|\leq \|x\|' \leq 2^C \|x\|. \]
			From these inequalities it is straightforward to compute the relation between the metric entropies.
		\end{proof}

		Subsets of $\RR^n$, and with it all finite dimensional vector spaces only have compact subsets of linear metric entropy.
		Infinite vector spaces, in contrast, always have bounded compact sets of arbitrary size.

		\begin{theorem}\label{resu:sets of arbitrary size}
			Whenever $(X,\norm\cdot)$ is an infinite dimensional normed space and $\mu:\omega\to\omega$ is a function, then there is a compact subset $K$ of $X$ such that $\size{K}\geq \mu$.
		\end{theorem}

		\begin{proof}
			By Riesz's Lemma there exists a sequence of elements $x_i$ such that $\|x_i\|=1$ and $\|x_i-x_j\|>\frac12$ for all $i\neq j$.
			This means that a closed ball of radius $1$ around zero contains all of the $x_i$, but a closed ball of radius $\frac 12$ around any point contains at most one $x_i$.
			Define a set $K\subseteq X$ by
			\[ K := \{0\} \cup\bigcup_{i\in\omega} \bigcup_{j=1}^{2^{\mu(i)}-2^{\mu(i-1)}}\{ 2^{-i+1}x_{j}\}, \]
			where we use the convention $2^{\mu(-1)}=0$.
			To cover $K$ with balls of radius $2^{-n}$ we need at least one ball for each of the $2^{\mu(n)}$ many elements $x\in K$ with $\|x\|\geq 2^{-n+1}$.
			Therefore $\big|K\big|(n)\geq \mu(n)$ for all $n\in\omega$.

			To argue that the constructed set $K$ is compact, show sequential compactness.
			For metric spaces this is equivalent to compactness.
			Consider an arbitrary sequence in $K$.
			From the construction it is clear that the ball of any radius around zero contains all but finitely many elements of $K$.
			By the pigeonhole principle there is either an element of $K$ that is visited infinitely often by the sequence, or there is an element of the sequence in any a ball of any radius around zero.
			In the first case there is a constant subsequence, in the second case there is a subsequence that is convergent to zero.
			Therefore the set is compact.
		\end{proof}

		It is well known that for Banach spaces finite dimensionality and local compactness coincide.
		A similar characterization follows from the above:

		\begin{corollary}
			For a normed space the following are equivalent:
			\begin{enumerate}
				\item It is infinite dimensional.
				\item For any $\nu:\omega\to\omega$ there exists a compact subset $K$ such that $\size{K}\geq \nu$.
				\item There exists a compact subset that has no linear modulus of total boundedness.
			\end{enumerate}
		\end{corollary}

		\begin{proof}
			\begin{itemize}
				\item[]
				\item[{\itshape1.}]$\Rightarrow${\itshape 2.}: This is \Cref{resu:sets of arbitrary size}.
				\item[{\itshape2.}]$\Rightarrow${\itshape 3.}: This is trivial.
				\item[{\itshape3.}]$\Rightarrow${\itshape1.}: By contradiction assume that the normed vector space was finite dimensional.
				Choose a basis to identify it with $\RR^d$.
				Since all norms on a finite dimensional vector space are equivalent, \Cref{resu:equivalent norms} shows that it is irrelevant which norm is used for whether or not all compact sets have a linear modulus of total boundedness.
				With respect to the supremum norm all sets have a linear modulus of total boundedness.
				This contradicts the third item.
			\end{itemize}
		\end{proof}

	\subsection{Metric entropy and complexity}\label{sec:sub:metric entropy and complexity}
		
		This chapter investigates connections between the concept of metric entropy and computational complexity.
		Recall the exhausting family $(K_l)_{l\in\omega^\omega}$ of compact subsets $K_l$ of the Baire space indexed by functions $l:\omega\to\omega$, namely:
		\[ K_l := \big\{\varphi\in\B \mid \forall n:\flength{\varphi}(n)\leq l(n) \big\} \]
		The following example is very instructive for the content of this section:		
		\begin{example}[A family of compacts]\label{ex:a family of compact sets}			
			Let $\M=(\Mn,d,(r_i))$ be a complete separable metric space with a dense sequence.
			For the Cauchy representation $\xi_{\M}$ of $\M$ from \Cref{def:Cauchy representation} it holds that
			\[ \xi_{\M}(K_l) = \bigcap_{n\in\omega} \bigcup_{i=0}^{2^{l(n)}-1} B^c_{2^{-n}}(r_i). \]
			The set on the right hand side is closed as an intersection of finite unions of closed sets and its metric entropy is bounded in terms of $l$.
			Due to the completeness of $\M$ the set is compact and for all $n\in\omega$ and non-decreasing $l:\omega\to\omega$
			\[ \size{\xi_{\M}(K_l)}(n) \leq l(n). \]
		\end{example}

		Surprisingly, a similar connection exists for arbitrary representations of metric spaces.
		The goal of this section is to obtain from a bound of the running time of the metric a modulus of total boundedness of the sets $\xi(K_l)$.
		For this we often consider the function that arises if the first order argument of a running time $T:\omega^\omega\times\omega\to \omega$ is fixed to some $l:\omega\to\omega$.
		We denote this function by $T(l,\cdot)$.
		\begin{theorem}[Complexity $\leadsto$ metric entropy]\label{resu:from complexity to metric entropy}
			Let $(\Nb,d)$ be a re\-pre\-sen\-ted metric space such that $d$ is computable in time $T$ on $\Nb\times \Nb$.
			Then there exists some $C\in\omega$ such that for all $l\in\omega^\omega$ the function $CT(l,\cdot +2)^2 + C$ is a modulus of total boundedness of $\xi_{\Nb}(K_l)$.
		\end{theorem}
		We may express this in brief by writing
		\[ \size{\xi_{\Nb}(K_l)} \in \bigo(T(l,\cdot+2)^2). \]
		Note however, that the above statement does not indicate the independence of the constant from the function $l$.
		The proof is postponed to the next section.
		For now consider some implications of this theorem and some examples.
		For instance:
		\begin{example}[\ref{ex:computable metric space}, \ref{ex:computable metric spaces continued}, \ref{ex:a family of compact sets} continued]
			Assume that $(\MM,d,(r_i))$ is a computable metric space, where the representation $\xi_{\MM}$ of $\MM$ is the Cauchy representation from \Cref{ex:computable metric space}.
			Assume that the discrete metric (cf. \Cref{eq:the discrete metric} on page \pageref{eq:the discrete metric}) is polynomial-time computable.
			According to \Cref{ex:computable metric spaces continued} the metric $d$ is computable in time polynomial in $l(n+2)$ and $n$.
			Thus, by \Cref{resu:from complexity to metric entropy} the metric entropy of $\xi_{\MM}(K_l)$ is bounded by a polynomial in $l(n+4)$ and $n$.
			This estimate is reasonably close to the estimate $l(n)$ for non-decreasing $l$ from \Cref{ex:a family of compact sets}.
		\end{example}

		Using Weihrauch's TTE and Weihrauch's representations (see \cite{MR1795407}) is equivalent to restricting to representations $\xi$ such that all names fulfill $\varphi(\str a) = \varphi(\sdone^{\length{\str a}})$ and $\length \varphi\equiv1$.
		We call such representations \demph{Cantor space representations}.
		For a Cantor space representation computability in time $T(l,n)$ is equivalent to computability in time $t:\omega\to\omega,$ $t(n) := T(1,n)$ within the TTE framework.
		\begin{corollary}
			Let $(\Nb,d)$ be a Cantor space represented metric space such that the metric is computable in time $t:\omega\to\omega$.
			Then $\size{\Nb}$ is bounded by a polynomial in $t(n+2)$.
		\end{corollary}
		Whenever there is a Cantor space representation of $C([0,1])$ such that the restriction of the evaluation operator to some set $F\subseteq C([0,1])$ can be computed in time $t(n)$, then the supremum norm on $F$ can be computed in time exponential in $t(n)$.
		This implies a weak version of the folklore fact from computable analysis that there is no Cantor space representation of the space of Lipschitz one functions on the unit interval that map zero to zero such that the evaluation is polynomial time computable (see for instance \cite{MR1952428}):
		In the above situation \Cref{resu:from complexity to metric entropy} specifies a bound on $\size{F}$ that is exponential in $t(n+2)$.
		This together with the fact that $C([0,1])$ as an infinite dimensional normed space has bounded subsets of arbitrary high metric entropy by \Cref{resu:sets of arbitrary size} shows that there is no Cantor space representation of $C([0,1])$ or even the unit ball of this space such that evaluation is computable in bounded time.
		More generally:
		\begin{corollary}
			No infinite dimensional normed space has a Cantor space representation such that the norm and the vector space operations are computable in bounded time when restricted to the closed unit ball.
		\end{corollary}
		
		\begin{proof}
			If the norm and the vector space operations are computable in bounded time, so is the metric.
			By the previous corollary this implies that there exists a modulus of total boundedness of the unit ball.
			Since the vector space is an infinite dimensional normed space, there exists arbitrary big compact subsets of the unit ball by \Cref{resu:sets of arbitrary size}.
			This is a contradiction, as moving to subsets increases the modulus of total boundedness by at most a shift by one.
		\end{proof}

		Recall from \Cref{resu:length of restrictions} that any range restriction of an admissible representation to a compact set has a length.
		The following corollary can be used to check whether a given representation has the minimal possible length.
		\begin{corollary}\label{resu:length of range restrictions of open representations}
			Let $\xi$ be an admissible representation of a metric space $(\Nn,d)$ such that the metric is computable in time $T$.
			Then the range restriction of $\xi$ to any compact subset $K\subseteq M$ has a length $\ell$ that fulfills
			\[ \size K \in \bigo(T(\ell,\cdot+2)^2). \]
		\end{corollary}

		\begin{proof}
			Openness is preserved under taking range restrictions.
			\Cref{resu:length of open representations} proves that the range restriction, as an open representation of a compact set, has a length $\ell$.
			Running times are also preserved under range restrictions.
			Apply \Cref{resu:from complexity to metric entropy} and get the assertion.
		\end{proof}

	\subsection{Proof of the main theorem}

		The goal of this section is to prove \Cref{resu:from complexity to metric entropy}.
		The argument works in a very general setting, in particular the assumption about the space can be weaker than computability of the metric in bounded time.
		This is not very surprising, as the metric entropy only mentions small balls and does not use any information about the exact values of the distance of points far away from each other.

		Equality in a metric space equipped with the Cauchy representation is usually not decidable.
		However, for two given elements $x,y$ it is verifiable if they are far apart.
		This can be formalized as computability of the following multivalued function:
		\begin{definition}\label{def:equality}
			Let $\N=(\Nb,d)$ be a represented metric space.
			Define the \demph{equality function} $\eq_{\N}:\Nb\times \Nb \mto \B$ by:
			\[ \eq_{\N}(x,y) := \left\{\varphi\in\binalbe^{\albe^*}\left|\begin{array}{c} \forall n: d(x,y) \leq2^{-n-1} \Rightarrow \varphi(\sdone^n) = \sdone\\
			\text{and} \\ \forall n:d(x,y) > 2^{-n} \Rightarrow \varphi(\sdone^n) = \sdzero
			\end{array}\right.\right\}. \]
		\end{definition}

		Informally but more intuitively one might indicate this by
		\[ \eq_{\N}(x,y)(\sdone^n) = \begin{cases} \sdone & \text{if } d(x,y)\leq 2^{-n-1} \\\sdzero & \text{if } d(x,y) > 2^{-n} \\ \sdzero\text{ or }\sdone &\text{otherwise.}\end{cases} \]

		\begin{definition}
			Let $\N=(\Nb,d)$ be a represented metric space.
			We say that \demph{equality is approximable in time $T$}, if its equality function $\eq_{\N}$ is computable in time $T$ on $\Nb\times \Nb$.
		\end{definition}
		
		This notion is indeed weaker than bounded time computability of the metric in a very precise sense:
		\begin{lemma}[From metric to equality]\label{resu:from metric to equality}
			Let $(\Nb,d)$ be a represented metric space.
			If $d$ is computable in time $T$ on $\Nb\times \Nb$, then equality is approximable in time $\bigo(\tilde T)$ for $\tilde T(l,n):= T(l,n+2)$.
		\end{lemma}

		\begin{proof}
			Let $M^?$ be a machine that computes $d$ in time bounded by $T$ on $\Nb\times \Nb$.
			Given a approximation requirement $\sdone^n$, carry out the computation this machine does on input $2^{n+2}-1$ to produce an encoding of a $2^{-n-2}$-approximation to $d(x,y)$.
			Return $\sdzero$ if the return value encodes a dyadic number strictly bigger than $2^{-n-1}+2^{-n-2}$, otherwise return $\sdone$.
			Carrying out the computations of $M$ takes time $T(l,n+2)$, and checking the return value against $2^{-n-1}+2^{-n-2}$ takes time less than twice that plus a constant.
			A triangle equality argument proves that this machine approximates the equality function.
		\end{proof}

		A running time bound restricts the access an oracle machine has to the oracle.
		The following lemma describes this dependence in exactly as much detail as needed for our purposes.
		It assigns to a machine a function $\B\times \albe^*\to \albe^*$ that takes as input an oracle and a string, and whose return value is a description of the communication between the machine and the oracle.
		In particular, if the first argument is changed the return values only change if the machine can distinguish the oracles in a computation with the second input as input string.
		
		Similar constructions have been done for oracles that only return $\sdzero$ or $\sdone$ before \cite{MR892102}.
		The function constructed is very closely related to moduli of sequentiality \cite{MR1911553}.
		Our approach differs from the ones taken in these sources in that a \lq dialog\rq\ only describes the return values of the oracle and not the queries.

		\begin{lemma}[Communication functions]\label{resu:oracle restriction}
			For any oracle Turing machine $M^?$ that runs in time $T$ on a set $A\subseteq \B$ there exists a function $L:A\times \Sigma^*\to\Sigma^*$ such that
			\begin{itemize}
				\item[(d)] $M^\varphi(\str a)$ is \underline{d}etermined by $L(\varphi,\str a)$, that is for all $\varphi,\psi\in A$ and strings $\str a$
				from $L(\varphi,\str a) =L(\psi,\str a)$ it follows that $M^\varphi(\str a) = M^\psi(\str a)$.
				\item[(v)] Each \underline{v}alue of an oracle on a string either matters a lot or does not matter at all:
				Whenever $\varphi,\psi,\phi\in A$ are string functions and $\str b$ is a string such that whenever $\str a$ is a string such that the $T(\flength\varphi,\length{\str b})$ initial segments of $\varphi(\str a)$ and $\psi(\str a)$ coincide, then so does the initial segment of $\phi(\str a)$ and if additionally $L(\varphi,\str b) = L(\psi,\str b)$ then it follows that $M^\varphi(\str b) = M^{\phi}(\str b)$.
				\item[(l)] The \underline{l}ength of $L$ is bounded by the running time:
				For all strings $\str a$ and $\varphi\in A$ 
				\[ \length{L(\varphi,\str a)}\leq 2\Big(T(|\varphi|,|\str a|)\cdot\big(T(|\varphi|,|\str a|) +1\big)+1\Big). \]
			\end{itemize}
		\end{lemma}

		Note that (v) implies (d), however, since the meaning of (d) is a lot easier to grasp and (v) is only needed to guarantee that pairings work as expected, they are stated separately.

		\begin{proof}
			Let $L(\varphi,\str a)$ be an encoding of the number of oracle queries together with a list of the $T(\flength{\varphi},\length{\str a})$ first bits of the answers to the oracle calls during the run $M^\varphi(\str a)$ of $M^?$ on input $\str a$ with oracle $\varphi$.
			I.e.\
			\[ L(\varphi)(\str a) = \langle N,\langle\str b_1,\ldots,\str b_N\rangle\rangle \]
			where $\str b_i$ consist of the $T(\flength{\varphi},\length{\str a})$ first bits of $\varphi(\str a_i)$ where $\str a_i$ is the $i$-th of the $N$ queries the machine asks to $\varphi$.

			As mentioned, the condition from (d) is implied by the one from (v).
			To see that the condition (v) holds note that the value $L(\varphi,\str b)$ determines the number $N$ of queries the machine asks the oracle $\psi$ and also their values $\str a_1,\ldots \str a_N$.
			Now $L(\varphi,\str b) = L(\psi,\str b)$ implies that the $T(\flength\varphi,\length{\str b})$ initial segments of $\varphi(\str a_i)$ and $\psi(\str a_i)$ coincide for all $i$.
			Therefore from the other assumption of (v) it follows that the run of $M^?$ on $\str b$ with oracle $\phi$ writes the same queries and gets answers that are indistinguishable for the machine.
			Thus $M^{\phi}(\str b)$ produces the same return value as both $M^\varphi(\str b)$ and $M^\psi(\str b)$.

			From the restriction of the running time of $M^?$ it follows that the number $N$ and each $\length{\str b_i}$ can at most be $T(|\varphi|,\length{\str a})$.
			This put together with the length estimations for the pairing functions from the introduction leads to the bound on the length of $L(\varphi,\str a)$.
		\end{proof}
		The conclusion $M^\varphi(\str b) = M^{\phi}(\str b)$ from item (v) cannot be replaced by the stronger $L(\varphi,\str b) = L(\phi,\str b)$.
		This is due to the use of initial segments of the oracle answers.
		The length of these initial segments depend on the value of the running time, which we have no control over.

		Recall that to each $l:\omega\to\omega$ a compact subset $K_l$ of the Baire space was assigned by
		\[ K_l:= \{\varphi\in \B\mid \flength{\varphi} \leq l\} \]
		and that the family $(K_l)_{l\in\omega^\omega}$ has the property that every compact subset of the Baire space is contained in some $K_l$.

		The proof of the following theorem is now a straightforward application of the previous lemma.
		Note that it does not require the metric space to be compact, but instead talks about certain relatively compact subsets of the space.

		\begin{theorem}\label{resu:from approximability to metric entropy}
			Let $\N$ be a represented metric space such that the equality is approximable in time $T$.
			Then for all $n\in\omega$
			\[ \size{\xi(K_l)}(n)\leq 2\big(T(l,n)\cdot(T(l,n)+1)+1\big). \]
		\end{theorem}

		\begin{proof}
			Let $\N=(\Nb,d)$ and fix some $n$.
			Let $M^?$ be the machine computing the equality function in time $T$ on the set
			\[ A:=\dom(\xi_{\Nb\times\Nb}) = \{\langle\varphi,\psi\rangle\in\B \mid \varphi\in\dom(\xi) \text{ and }\psi\in \dom(\xi)\}, \]
			and let $L:A\times \Sigma^*\to\Sigma^*$ be the communication function assigned to $M^?$ by \Cref{resu:oracle restriction}.
			Let $I$ be the set of strings $\str a$ such that there exists a $\psi\in K_l\cap\dom(\xi)$ such that $L(\langle\psi,\psi\rangle,\sdone^{n+1}) = \str a$.
			For each $i\in I$ choose some $\psi_i\in K_l\cap \dom(\xi)$ such that $L(\langle\psi_i,\psi_i\rangle,\sdone^{n+1}) = i$.
			From the size limit for $L(\psi,\sdone^{n+1})$ from \Cref{resu:oracle restriction} item (l) it follows that
			\[ \#I \leq 2^{2(T(l,n)\cdot(T(l,n)+1)+1)}. \]

			Claim that the closed $2^{-n}$-balls around the $\xi(\psi_i)$ cover $\xi(K_l)$:
			Indeed, take an arbitrary $x\in\xi(K_l)$, that is $x=\xi(\psi)$ for some $\psi\in K_l \cap \dom(\xi)$.
			Then for $i:= L(\langle\psi,\psi\rangle,\sdone^n)\in I$ it holds that $L(\langle\psi_i,\psi_i\rangle,\sdone^n) = L(\langle\psi,\psi\rangle,\sdone^n)$.
			Use the property of $L$ from item (v) of \Cref{resu:oracle restriction} for the functions $\langle\psi,\psi\rangle$, $\langle \psi_i,\psi_i\rangle$ and $\langle\psi_i,\psi\rangle$.
			To do so, it is necessary to show that if certain beginning segments of $\langle\psi,\psi\rangle(\str b)$ and $\langle\psi_i,\psi_i\rangle(\str b)$ coincide, then so does $\langle \psi_i,\psi\rangle(\str b)$.
			This is true for any initial segments: If $\langle\psi,\psi\rangle(\str b) = \langle\psi(\str b),\psi(\str b)\rangle$ and $\langle \psi_i(\str b),\psi_i(\str b)\rangle$ coincide, the corresponding padded versions of $\psi(\str b)$ and $\psi_i(\str b)$ coincide and one may as well only swap one of the padded strings.
			It follows that
			\[ M^{\langle\psi_i,\psi\rangle}(\sdone^n) = M^{\langle\psi,\psi\rangle}(\sdone^n) = \sdone. \]
			Since $M^?$ computes the function $\mathrm{eq}_{\N}$ from \Cref{def:equality}, this implies that $d(x,\xi(\psi_i)) \leq 2^{-n}$.
			Thus, $x\in B^c_{2^{-n}}(\xi(\psi_i))$ and, since $x\in \xi(K_l)$ was arbitrary, the closed $2^{-n}$-balls around the images of the $\psi_i$ cover $\xi(K_l)$.
		\end{proof}

		Finally, let us argue that the theorem indeed implies \Cref{resu:from complexity to metric entropy}.

		\begin{proof}[of \Cref{resu:from complexity to metric entropy}]
			Let $\N=(\Nb,d)$ be a represented metric space and assume that $d$ is computable in time $T$ on $\Nb\times\Nb$.
			\Cref{resu:from metric to equality} provides a constant $\tilde C\in\omega$ such that equality on $\N$ is approximable in time $(l,n)\mapsto \tilde C T(l,n+2)+\tilde C$.
			Apply \Cref{resu:from approximability to metric entropy} to obtain
			\[ \size{\xi(K_l)}(n)\leq 2\big((\tilde CT(l,n+2) +\tilde C)(\tilde T(l,n+2) + \tilde C+1)+1\big)\in\bigo(T(l,n+2)^2), \]
			where the constant $C$ can be chosen as $8\tilde C^2+4\tilde C+2$.
			The independence of $l$ follows from the independence of $\tilde C$ of $l$ (by definition of $\bigo(T)$).
		\end{proof}


\section{Regular and complete representations}\label{sec:regularity and completeness}

	Some functions $T:\omega^\omega \times \omega\to \omega$ are no reasonable candidates for running times.
	A running time should grant more time, if the input is bigger; therefore it should be monotone.
	We use the notion of monotonicity introduced by Howard \cite{MR0469712,MR2130066}.
	It restricts to monotonicity in the point-wise sense, if the function is only considered on increasing inputs and returns only increasing functions.
	\begin{definition}
		A function $T:\omega^\omega\times \omega\to\omega$ is \demph{monotone}, if for all $l,l'\in\omega^\omega$ from $\forall n\leq m:l(n)\leq l'(m)$ it follows that $\forall n\leq m:T(l,n)\leq T(l',m)$.
	\end{definition}

	In \Cref{sec:construction of standard representations} we encounter a situation, where a machine needs to compute a function that is similar to its running time.
	To guarantee that this can be done by the machine without taking too many steps, a notion of time-constructibility is needed for second-order running times. 
	\begin{definition}\label{def:running time}\label{def:time-constructible}
		We call $T:\omega^\omega\times \omega \to \omega$ \demph{time-constructible on $A\subseteq \B$}, if there is an oracle Turing machine that computes the mapping $(\varphi,\str a)\mapsto T(\flength \varphi,\length{\str a})$ and runs in time $\bigo(T)$ on $A$.
	\end{definition}
	We call a function \demph{time-constructible}, if it is time-constructible on the Baire space.
	This definition reproduces usual notion of time-constructibility for functions that are independent of the first argument.

	\begin{example}\label{ex:length function}
		Recall that the length function $\flength\cdot:\B\to\omega^\omega$ was defined by $\flength{\varphi}(n) := \max\big\{\length{\varphi(\str a)}\,\big|\, \length{\str a}\leq n\big\}$.
		Due to the evaluation of the maximum taking an exponential number of oracle queries, this function is not polynomial-time computable.
		Thus, the function $(l,n) \mapsto l(n)$ is not time-constructible.
		As a consequence, most second-order polynomials are not time-constructible.
		The function $(l,n)\mapsto 2^{\max\{l(n),n\}}$, on the other hand, is time-constructible.
	\end{example}

	The standard application of time-constructibility is increasing the domain of realizers:
	\begin{lemma}
		Let $F:\subseteq \B\to \B$ be a functional on the Baire space computable in time $T$ on some set $A$.
		If $T$ is time-constructible on $B\supseteq A$, then $F$ is computable in time $\bigo(T)$ on $B$.
	\end{lemma}

	\subsection{Regularity}\label{sec:regularity}

		The mapping $L(l,n):=l(n)$ plays a special role within the second-order polynomials.
		Indeed, if $L$ is time-constructible on a set $A$, then many second-order polynomials are time-constructible on $A$.
		Under additional assumptions about $A$, for instance that all elements are strictly increasing it is possible to prove that time-constructibility of $L$ on $A$ implies time-constructibility of all second-order polynomials on $A$.
		\begin{definition}\label{def:linear-time computability}
			We call a functional $F\colon\B\to\B$ \demph{linear-time computable}, resp.\ \demph{linear-time computable on} $A\subseteq B$, if it is computable in time $\bigo(L)$ for $L(l,n):=l(n)$ resp.\ computable in time $\bigo(L)$ on $A$.
		\end{definition}
		Using this concept the content of \Cref{ex:length function} can be formulated as \lq $L$ is time-constructible on $A$ if and only if the length function is linear-time computable on $A$\rq.

		If it is impossible to find an upper bound on the length of a name in a representation this will usually lead to difficulties in applications.
		The following regularity condition on representations removes most of these difficulties, while not restricting the freedom of choice of representations too much:
		\begin{definition}\label{def:regular}
			We call a representation $\xi$ \demph{regular}, if there exists a linear-time computable upper bound of the length function.
			I.e. if there exists a linear-time computable function $m:\B\to\omega^\omega$ such that for all $\varphi \in\dom(\xi)$ it holds that $m(\varphi)$ is monotone and $\length\varphi\leq m(\varphi)$.
		\end{definition}
		If the length function is linear time computable on $\dom(\xi)$, then $\xi$ is regular.
		Due to the use of an upper bound the converse does not hold.
		However, from the linear time computability of $m$ it follows that there exists some $C\in\omega$ such that $m(n)\leq C\flength\varphi(n) +C$.
		
		The same problem has been tackled by Kawamura and Cook in a slightly less general way by introducing a restricted class of string functions that are allowed to be names (cf. \cite{Kawamura:2012:CTO:2189778.2189780} and \cite{kawamuraphd}).
		\begin{definition}
			A string function $\varphi\in\B$ is called \demph{length-monotone}, if for all strings $\str a,\str b$ it holds that
			\[  \length{\str a}\leq \length{\str b} \quad\Rightarrow\quad\length{\varphi(\str a)} \leq \length{\varphi(\str b)}. \]
			The set of length-monotone string functions is denoted by $\reg$.
		\end{definition}
		For length-monotone string functions it holds that $\flength\varphi(n) = \length{\varphi(\sdone^n)}$.
		Therefore, the length function is linear-time computable on $\Reg$ and all second-order polynomials are time-constructible on $\Reg$.

		\begin{definition}
			A representation is called a \demph{second-order representation}, if its domain is contained in $\Reg$.
		\end{definition}
		The use of the term \lq second-order\rq\ indicates applicability of second-order complexity theory, not the use of higher-order objects in the representation.
		The restriction to length-monotone names leads to excessive padding and technical complications that are avoided by regular representations.
		Since the length function is time constructible on $\Reg$ any second-order representation is regular.
		More generally:

		\begin{proposition}[Regularity vs. second-order]\label{resu:regularity vs. second-order}
			For every regular representation there is a linear-time equivalent second-order representation.
			Second-order re\-pre\-sen\-ta\-ti\-ons are regular.
		\end{proposition}

		\begin{proof}
			Let $\xi$ be a second-order representation.
			Set $m(\varphi)(n):= \length{\varphi(\sdone^n)}$.
			This function is obviously linear-time computable and since any element of $\dom(\xi)$ is length-monotone and therefore fulfills $\flength{\varphi}(n) = \length{\varphi(\sdone^n)}$, it is not only an upper bound to the length function, but restricts to the length function on $\dom(\xi)$.

			On the other hand assume that $\xi'$ is a regular representation.
			Let $m:\B\to\omega^\omega$ be the linear time computable upper bound on the length function.
			Define a padding function as follows:
			For a string $\str a$ denote by $\tilde{\str a}$ the string where each $\sdzero$ in is replaced by $\sdzero\sdone$ and each $\sdone$ is replaced by $\sdone\sdone$.
			Set
			\[ \mathrm{pad}(\varphi)(\str a) := \widetilde{\varphi(\str a)}\sdzero^{2\max\{m(\varphi)(\length{\str a}) - \length{\varphi(\str a)},0\}}. \]
			Since $m(\varphi)$ is increasing and $m$ is an upper bound to the length function, $\mathrm{pad}(\varphi)$ is a length-monotone function, whenever $\varphi$ is an element of $\dom(\xi')$.
			Since $m$ is linear time computable, the mapping $\mathrm{pad}$ is linear time computable.
			Thus the representation $\xi'$ can be translated to the second-order representation $\xi^{\mathrm{pad}}$ in linear time, where $\varphi$ is a $\xi^{\mathrm{pad}}$-name of an element if and only if $\varphi=\mathrm{pad}(\psi)$ for a $\xi'$-name $\psi$ of the element.

			A linear time translation in the opposite direction is easily written down by first removing pairs of $\sdzero$s from the return-values of a string function and then undoing the encoding of single digits by pairs of digits.
		\end{proof}

	\subsection{Completeness}
		Usually bounded time computability is a more restrictive notion than computability.
		However, for total functions on complete spaces one might expect any computable function to be computable in bounded time.
		This is for instance true for the real numbers represented as in \Cref{def:standard representation of the reals}.

		\begin{lemma}
			Let $A\subseteq \B$ be closed.
			Any computable total function $F:A\to \B$ is computable in bounded time on $A$.
			\begin{proof}
				Let $M^?$ be an oracle Machine computing the total function.
				Consider the function $\Time_{M^?}:A\times \Sigma^*\to \omega$, where $\Time_{M^?}(\varphi,\str a)$ is the number of steps in the run $M^\varphi(\str a)$.
				This function is continuous.
				To see this let $(\varphi_n,\str a_n)\subseteq A$ be a sequence converging to $(\varphi,\str a)\in A$.
				Since $\Sigma^*$ is equipped with the discrete topology, there exists some $N$ such that for all $n \geq N$ it holds that $\str a_n = \str a$.
				Furthermore, the run $M^\varphi(\str a)$ takes a finite number $k$ of steps.
				Let $\tilde N\geq N$ be such that the initial segments of length $k$ of the $\varphi_n$ coincide whenever $n\geq \tilde N$.
				Since $M^?$ is a deterministic machine, the runs $M^\varphi(\str a)$ and $M^{\varphi_n}(\str a_n)$ coincide whenever $n$ is bigger than $\tilde N$.
				Thus, the sequence $\Time_{M^?}(\varphi_n,\str a_n)$ is eventually constant and converges in $\omega$.
				
				Since $A$ is closed, $K_l\cap A$ is compact and since $\Time_{M^?}$ is continuous, the number
				\[ T(l,n) := \max\big\{\Time_{M^?}(\varphi,\str a) \mid\varphi\in K_l\text{ and }\length{\str a} \leq n\big\} \]
				exists as the maximum of a continuous function on a compact set.
				By the definition of $\Time_{M^?}$ the machine $M^?$ runs in time $T$.
				Thus, the function is computable in bounded time.
			\end{proof}
		\end{lemma}
		In particular any computable total function $F:\B\to\B$ is computable in bounded time.

		A condition that guarantees that each computable function has a total computable realizer is the following:
		Recall that a subset $A$ of the Baire space is called \demph{co-recursively enumerable closed} or \demph{co-r.e.\ closed}, if there exists an oracle Turing machine $M^?$ such that the run $M^\varphi(\varepsilon)$ terminates if and only if $\varphi\notin A$.
		\begin{definition}
			We call a representation \demph{complete}, if its domain is closed, and \demph{computably complete}, if its domain is co-r.e.\ closed.
		\end{definition}
		All of the representations we construct are complete.
		Typical examples of representations that are not complete are the \lq padded\rq\ counterexample representations constructed by Kawamura and Pauly to prove that a straight-forward generalization of admissibility to the polynomial-time framework does not provide a reasonable notion \cite[Section 7]{MR3219039}.

		\begin{lemma}
			Let the representation of $\XX$ be complete.
			Then any computable function with domain $\XX$ has a total computable realizer.
		\end{lemma}

		\begin{proof}
			Let $M^?$ be a machine that computes a realizer of the function $f$.
			Let $\tilde M^?$ be a machine that witnesses the computable closedness of the domain of $\xi$.
			Construct an oracle Turing machine $N^?$ as follows:
			When given access to an oracle $\varphi$ and an input string $\str a$, it dovetails simulations of $M^\varphi(\str a)$ and $\tilde M^\varphi(\varepsilon)$.
			If the simulation of $M^\varphi(\str a)$ terminates first, the machine returns the return value of this computation.
			If the simulation of $\tilde M^\varphi(\varepsilon)$ terminates first, the machine stops and returns the empty string.
			On one hand, since $M^?$ computes a realizer of a function, $M^\varphi(\str a)$ terminates whenever $\varphi$ is an element of $\dom(\xi)$.
			On the other hand if $\varphi\notin \dom(\xi)$ then $\tilde M^\varphi(\varepsilon)$ is guaranteed to terminate.
			Thus $N^?$ indeed computes a total function on the Baire space.
			It is clear from its definition that this function is still a realizer of $f$.
		\end{proof}

		The following is the combination of the previous two lemmas:

		\begin{theorem}\label{resu:completeness as properness}
			Let $f$ be a computable function with domain $\XX$.
			If the representation of $\XX$ is complete, then $f$ is computable in time bounded on $\XX$.
			If it is computably complete, then $f$ is computable in bounded time.
		\end{theorem}
		An example of a situation where provably no reasonable complete representations exist is provided by Férée and Hoyrup:
		They prove that for a represented non $\sigma$-compact Polish space $\XX$ there is no representation of $C(\XX)$ such that the time complexity of the evaluation is well defined (\cite[Theorem 3.1]{higherorder}, a full proof can be found in the PhD thesis by Férée).
		By the above theorem, in this case there can not  exists a complete representation of $C(\XX)$ such that the evaluation is computable.
		This is not surprising, as in fact it is known that the descriptive complexity of the domain of representations increases with the number of applications of the function space construction \cite{MR3417080}.

		The following can be understood as the motivation for the names \lq complete\rq\ resp.\ \lq computably complete\rq.

		\begin{theorem}[Completenesses]\label{resu:completenesses}
			A Cauchy representation of a separable metric space is complete if and only if the space is complete.
			The Cauchy representation of a complete computable metric space is computably complete.
		\end{theorem}

		\begin{proof}
			Let $\M=(\Mn,d,(r_i))$ be the metric space with dense sequence.
			First assume that $\M$ is complete and let $\varphi_n$ be a sequence of names of elements $x_n$ that converges to an element $\varphi$ of the Baire space.
			To see that $\varphi$ is a name, first note that the sequence $(x_n)$ is a Cauchy sequence:
			Given $n$ choose $N$ such that for all $k\geq N$ it holds that $\varphi_k(2n+1) = \varphi(2n+1)$.
			And therefore for $k,l\geq N$
			\[ d(x_k,x_l) \leq d(x_k,r_{\varphi_k(2n+1)})+d(r_{\varphi_l(2n+1)},x_l) \leq \frac1{n+1}. \]
			Since the metric space is complete and $x_n$ is a Cauchy sequence, there exists a limit $x$ of the Cauchy sequence.
			That $\varphi$ is a name of $x$ follows from taking the limit $m\to \infty$ in
			\[ d(x,r_{\varphi(n)}) \leq d(x,x_m) + d(x_m,r_{\varphi(n)}). \]
			Note that $\varphi(n) = \varphi_m(n)$ for $m$ big enough and $d(x_m,r_{\varphi_m(n)})\leq \frac1{n+1}$.

			Now assume that the domain of the Cauchy representation is closed and let $x_n$ be a sequence of elements in $\M$.
			From the Cauchy property it follows that it is possible to find a convergent sequence of names $\varphi_n$.
			Since the domain of the representation is closed the limit $\varphi$ is the name of some $x\in\Mn$, and since the Cauchy representation is continuous, $x$ is the limit of the sequence $x_n$.
			Therefore $\M$ is complete.

			Finally assume that $\M$ is a computable metric space and let $\xi$ be the Cauchy representation.
			To see that the domain of the Cauchy representation is co-r.e.\ first prove that
			\[ \varphi \in\dom(\xi)\quad \Leftrightarrow \quad \forall i,j\in\NN :d(r_{\varphi(i)},r_{\varphi(j)})\leq \frac1{i+1}+\frac1{j+1}. \]
			For the first implication assume that $\varphi$ is a name of some element $x$ in the Cauchy representation.
			Thus, using the triangle inequality and the symmetry of the metric see that for all $i,j\in\NN$
			\[ d(r_{\varphi(i)},r_{\varphi(j)})\leq d(r_{\varphi(i)},x) + d(x,r_{\varphi(j)}) \leq \frac1{i+1}+\frac1{j+1}. \]
			For the other implication assume that $\varphi$ is such that the right hand side of the equivalence is fulfilled.
			It follows that $(r_{\varphi(i)})$ is a Cauchy sequence.
			Since $\M$ was assumed to be complete, this Cauchy sequence converges to some element $x$ and for any $i,j\in\NN$
			\[ d(x,r_{\varphi(i)}) \leq d(x,r_{\varphi(j)}) + d(r_{\varphi(j)},r_{\varphi(i)}) \leq d(x,r_{\varphi(j)}) + \frac1{j+1}+\frac1{i+1}. \]
			Taking the limit $j\to \infty$ in this inequality proves that $d(x,r_{\varphi(i)})\leq \frac1{i+1}$ for all $i\in\NN$, i.e.\ that $\varphi$ is a name of $x$ in the Cauchy representation.
			
			Thus, a machine that correctly identifies string functions that are not names in the Cauchy representation can be constructed by searching for values $i$ and $j$ such that $d(r_{\varphi(i)},r_{\varphi(j)}) > \frac1{i+1}+\frac1{j+1}$.
			This search can be done by a machine since the discrete metric is computable.
		\end{proof}

		Since completeness is not preserved under any kind of equivalence, the proof of the above theorem is repeated several times throughout the rest of the paper.

\section{Construction of standard representations}\label{sec:construction of standard representations}
	
	This chapter specifies families of representations for which it is possible to prove that the estimate obtained from \Cref{resu:from complexity to metric entropy} is close to optimal.
	This requires two steps.
	The first step is to eliminate the inherent computational difficulty of the metric, as this can not be reflected in the metric structure of the space.
	Recall from \Cref{ex:computable metric spaces continued} that the bound on the metric entropy obtained from the Cauchy representation by \Cref{resu:from complexity to metric entropy} improves with the running time of the discrete metric $\tilde d$.
	Within the framework of representations it can be simulated that the discrete metric has the lowest possible running time by considering a representation that additionally provides the discrete metric as oracle.
	\begin{definition}\label{def:the relativized cauchy representation}
		Let $\M=(\Mn,d,(r_i))$ be a complete separable metric space with dense sequence.
		Define the \demph{relativized Cauchy representation} $\xi_{\M}^r$ as follows:
		$\varphi\in\B$ is a $\xi_{\M}^r$-name of $x\in \Mn$ if and only if for all $n\in\NN$ the string $\varphi(\sdzero n)$ encodes a non-negative integer $i$ with $d(x,r_i)\leq \frac1{n+1}$ and for all $k,m,n\in\NN$ 
		\[ \abs{d(r_{k},r_m)-\frac{\varphi(\sdone\langle k,m,n\rangle)}{n+1}}\leq \frac1{n+1}.  \]
	\end{definition}
	That is: Any name comes with an oracle for the discrete metric $\tilde d$.
	Note that $\M$ was not assumed to be a computable metric space.
	In particular the discrete metric may not be computable.
	If it is incomputable then $\xi^r_{\M}$ does not have any computable string functions in its domain.
	In any case, $\xi^r_{\M}$ renders the metric computable in time $\bigo(T)$ for $T(l,n):=\max\{l(n+2),n\}$ and the bound obtained by \Cref{resu:from complexity to metric entropy} is $\size{\xi^r_{\M}(K_l)}\leq C l(n+2)^2+C$.
	This is reasonably close to the bound $\size{\xi_{\M}(K_l)} \leq l(n)$ that came from the concrete structure of the Cauchy representation (cf. \Cref{ex:a family of compact sets}).
	This bound itself, however, may still be far off, as the one point space as example shows.

	\Cref{resu:completenesses} remains valid if the Cauchy representation is replaced by the relativized Cauchy representation:
	\begin{proposition}[Relativized completenesses]\label{resu:relativized completenesses}
		A relativized Cauchy representation of a separable metric space is complete if and only if the space is complete.
		The relativized Cauchy representation of a computable complete metric space is computably complete.
	\end{proposition}
	\begin{proof}
		If the relativized Cauchy representation is complete, it follows that the Cauchy representation with respect to the same dense sequence is complete, thus completeness of the metric space follows from \Cref{resu:completenesses}.

		Now assume that the space is complete and let $\varphi_n$ be a sequence of names that converges to some $\varphi\in\B$.
		The proof of \Cref{resu:completenesses} can be copied to see that the values $\varphi(\sdzero n)$ encode indices of approximations to an element.
		That the values of $\varphi(\sdone\langle k,m,n\rangle)$ are valid approximations of the discrete metric follows from the use of non-strict inequality.

		Finally assume that the metric space is computably complete.
		A machine that terminates on the empty string if and only if the oracle is not a name can be described as follows:
		It dovetails the procedure described in the proof of \Cref{resu:completenesses} and a search for incompatibilities between the values of the discrete metric and values of the names when they are interpreted as an oracle for the discrete metric.
	\end{proof}

	Note that while the relativized Cauchy representation is a good starting point, for many applications from complexity theory it is not appropriate.
	For instance the standard representation of Lipschitz one functions that map zero to itself does not polynomial-time translate to a Cauchy representation.
	The examples right before the conclusion illustrate that the constructions below include these as special cases.
	It is not clear how to tackle the problem in the most general case of a separable metric space, thus we present two results for restricted classes of spaces:
	First compact metric spaces and then separable Banach spaces.

	\subsection{Compact metric spaces}

		The basic idea for compact metric spaces is to use a refinement of the relativized Cauchy representation for a sequence that is well-behaved in the following sense:
		\begin{definition}\label{def:uniformly dense}
			A sequence $(r_i)$ in a compact metric space $(\Nn,d)$ is called \demph{uniformly dense}, if both of the following hold:
			\begin{itemize}
			 	\item[(c):] $(r_i)$ has the \underline{c}overing property:
			 	The closed $2^{-n}$-balls around the first $2^{\size{\Nn}(n)+\lceil\lb(n+1)\rceil}$ elements cover $\Nn$.
			 	\item[(s):] $(r_i)$ has the \underline{s}panning property:
			 	For any $k\leq \size{\Nn}(n-1)$ there are at least $\lceil2^{k-1}\rceil$ elements of pairwise distance strictly more than $2^{-n}$ within $r_0,\ldots, r_{2^{k}-1}$.
			\end{itemize}
		\end{definition}
		Examples for uniformly dense sequences are the standard enumeration of the dyadic numbers in the unit interval (cf. \Cref{ex:the unit interval}) or standard enumerations of the piece-wise continuous functions with dyadic breakpoints in the Lipschitz one functions that map zero to zero.

		Under the assumption of compactness, uniformly dense sequences do always exists:
		\begin{lemma}[Uniformly dense sequences]\label{resu:uniformly dense}
			In every infinite compact metric space there is a uniformly dense sequence.
		\end{lemma}

		\begin{proof}
			Let $(\Nn,d)$ be the compact metric spaces.
			Let $I_n$ be a maximal set of elements $y_{i,n}$ of pairwise distance strictly more than $2^{-n+1}$.
			From the maximality it can be seen that the closed $2^{-n}$-balls around the $y_{i,n}$ cover $\Nn$ and therefore $\#I_n\leq 2^{\size{\Nn}(n)}$.
			Let $(r_i)$ be the sequence that arises by writing the tuples after one another.

			To verify that $(r_i)$ has the covering property it suffices to note that all elements of $I_n$ are listed within the first
			\[ \sum_{m=0}^n \# I_m \leq \sum_{m=0}^n 2^{\size{\Nn}(m)} \leq (n+1)2^{\size{\Nn}(n)} \leq 2^{\size{\Nn(n)}+\lceil\lb(n+1)\rceil} \]
			elements.
			
			The spanning property follows by induction:
			For a fixed $m< n-1$ at most one $y_{j,m+1}$ can lie in the closed $2^{-n}$-ball around $y_{i,m}$.
			Thus, starting from the beginning and always only adding those elements that are far enough away from the ones previously chosen, one always ends up with at least $2^{k-1}$ elements.
		\end{proof}

		Using this, the main theorem of this section can be proven:

		\begin{theorem}[Representing compact spaces]\label{resu:representing compact spaces}
			Let $(\Nn,d)$ be an infinite compact metric space.
			Let $\ell:\omega\to \omega$ be a function such that there exists a time-constructible and monotone function $S:\omega^\omega\times\omega \to \omega$ such that
			\begin{equation} \label{eq:condition time} \ell(n)S(\ell,n) \geq \size{\Nn}(n) + \lceil \lb(n+1)\rceil. \end{equation}
			Then there exists an admissible, regular, complete representation $\xi$ of length $\ell$ such that the metric is computable in time $\bigo(T)$ for
			\[ T(l,n) := l(n+2)\cdot S(l,n+2) \]
			and for all strictly monotone $l:\omega\to \omega$ and $n\in\omega$
			\[ l(n)\cdot S(l,n)\leq \size \Nn(n-1) \Rightarrow l(n)\cdot S(l,n) \leq \size{\xi(K_l)}(n+1)+1. \]
		\end{theorem}
		
		\begin{proof}
			Let $(r_i)_{i\in\NN}$ be a uniformly dense sequence like the one constructed in \Cref{resu:uniformly dense}.
			Define $\xi$ as follows:
			A string function $\varphi$ is a name of an element $x\in \Nn$ if and only if the following conditions hold:
			\begin{itemize}
				\item[(l)] $\varphi$ provides its \underline{l}ength: For all $n\in\NN$ it holds that $\length{\varphi(\sdzero^n)} = \flength{\varphi}(n)$.
				\item[(a)] $\varphi$ provides indices of \underline{a}pproximations:
				For all $n\in\NN$ the strings $\varphi(\sdzero\langle j,n\rangle)$ where $j$ reaches from $0$ to $S(\flength\varphi,\length{n})$ have length less than $\flength\varphi(\length{n})$ and their concatenation is a non-negative integer $i$ with $d(x,r_i) \leq \frac1{n+1}$.
				\item[(o)] $\varphi$ provides an \underline{o}racle for the metric: For any $i,j,n\in\NN$ it holds that $\varphi(\sdone\langle i,j,n\rangle)\in\ZZ$ and
				\[ \abs{d(r_i,r_j)-\frac{\varphi(\sdone \langle i,j,n\rangle)}{n+1}}\leq \frac1{n+1}. \]
			\end{itemize}

			This defines a second-order representation:
			For any two distinct elements $x,y\in \Nn$ there is an $n$ such that $2^{-n}<d(x,y)$.
			Thus, if $r_k$ resp.\ $r_m$ are $2^{-n-1}$-approximations of $x$ resp.\ $y$, then $k\neq m$.
			Now assume that $\varphi$ and $\psi$ fulfill the conditions to be names of $x$ resp.\ $y$.
			Then the indices from (a) must differ.
			Thus, $\xi$ is single-valued.
			The choice of the sequence $(r_i)_{i\in\NN}$ and \eqref{eq:condition time} make sure that condition (a) can be fulfilled by a function of length $\ell$, thereby leaving enough freedom in the choice of the function to make it also fulfill (m).
			Thus, $\xi$ has length $\ell$ and is in particular surjective.

			It is left to provide an appropriate algorithm for computing the metric.
			When given a name $\langle \varphi,\psi\rangle$ of some element $(x,y)\in (\Nn,\xi)\times(\Nn,\xi)$ as oracle and a precision requirement $n$ this algorithm proceeds as follows:
			From the time-constructibility of $S$ it follows that $N:=S(\flength{\langle\varphi,\psi\rangle},\length{n}+2)$
			can be computed in time $\bigo(T)$.
			Next the machine queries the oracle $N$ times for $\sdzero\langle j,4n+3\rangle$, with the values of $j$ going from $0$ to $N$ to obtain indices $i$ and $k$ of $\frac1{4n+3}$-approximations to $x$ and $y$.
			This takes time less than $\bigo(N \cdot \flength{\langle\varphi,\psi\rangle}(\length{n}+2))$.
			Finally, the machine queries the oracle input $\sdone\langle i,k,2n+1\rangle$ for an approximation of $d(r_i,r_k)$.
			Using the triangle inequality one verifies that the result leads to a valid dyadic approximation to $d(x,y)$ being written on the output tape.

			The lower bound on the metric entropy follows from the spanning property of the uniformly dense sequence (see \cref{def:uniformly dense}).
			The admissibility follows since it is possible to translate back and forth between the representation $\xi$ and the relativized Cauchy representation in time $\bigo(S)$.
			For the completeness of the representation use that compactness of a metric space implies completeness and that the time constructibility of $S$ implies continuity.
			The latter can be used to prove that the images of a converting sequence of names forms a Cauchy sequence and to repeat the argument from the proof of \Cref{resu:completenesses} that the limit of a convergent sequence of names is a name of the limit.
		\end{proof}

		The representation depends heavily on the choice of the uniformly dense sequence.
		Note that there is no computability condition on the compact metric space.
		In particular there is no guarantee that the representation has any computable names at all.

	\subsection{Separable Banach spaces}

		The goal of this chapter is to obtain a similar result as in the previous section for the case where the metric space is not compact but instead a Banach space.
		The construction is inspired by approximation theory:
		The representation is chosen in such a way that the images of the sets  $K_l$ are similar to full approximation sets as introduced in \cite{lorentz1966} (cf. also \Cref{sec:full approximation sets}).

		\begin{definition}\label{def:fundamental system}
			Let $\XX=(X,\norm\cdot)$ be a Banach space.
			A sequence $(e_i)$ in $\XX$ is called a \demph{fundamental system} if $\norm{e_i}=1$ for all $i$ and the linear span of the $e_i$ is dense in $\XX$.
			That is: For each $x\in\XX$ and $n\in\NN$ there exists an $N\in\NN$ and some real numbers $\lambda_1,\ldots,\lambda_N$ such that
			\[ \norm{\sum_{i=1}^N\lambda_i e_i - x} \leq 2^{n}. \]
		\end{definition}
		A fundamental system should be regarded as the Banach space equivalent of a dense sequence in a compact metric space.
		Any separable Banach space has a fundamental system as a dense sequence can be normalized.

		The following notion can be understood as a stronger version of the equivalent of being a uniformly dense sequence in a compact metric space:
		\begin{definition}\label{def:schauder basis}
			Let $\XX$ be a Banach space.
			A fundamental system $(e_i)_{i\in\NN}$ is called a \demph{Schauder basis}, if for each $x\in \XX$ there exists exactly one sequence $(\lambda_i)_{i\in\NN}$ of real numbers such that
			\[ x = \lim_{n\to \infty} \sum_{i=0}^n \lambda_i e_i. \]
		\end{definition}
		However, in contrast to the situation for compact metric spaces there is no way to construct a Schauder basis in an arbitrary separable Banach space.
		Quite to the contrary of \Cref{resu:uniformly dense}, that constructs a uniformly dense sequence in any compact metric space, a famous example by Enflo construct a separable Banach space that does not have any Schauder basis \cite{MR0402468}.

		Recall that the norm of a linear continuous functional $f$ on a Banach space $\XX$ is defined by
		\[ \norm f := \sup_{\norm x\leq 1}\{\abs{f(x)}\}. \]
		We need the following well-known result for Schauder bases that follows directly from the Uniform Boundedness Principle:
		\begin{theorem}\label{resu:uniform boundedness principle}
			Let $(e_i)_{i\in\NN}$ be a Schauder basis in a Banach space $\XX$.
			Then there are continuous linear functionals $(f_i)_{i\in\NN}$ such that for all $x$ in $\XX$ it holds that
			\[ x = \lim_{n\to \infty} \sum_{i=0}^n f_i(x) e_i. \]
			Furthermore, there is a constant $C$ such that
			\[ \forall i\in\NN:\norm{f_i} \leq C. \]
		\end{theorem}

		The consequence we need is that there is a common lower bound on the pairwise distance of the elements of a Schauder basis.
		\begin{corollary}\label{resu:riezs property of schauder bases}
			Let $(e_i)_{i\in\NN}$ be a Schauder basis in a Banach space.
			Then there exists a constant $\alpha>0$ such that for all non-negative integers $i\neq j$
			\[ \norm{e_i-e_j}> \alpha. \]
		\end{corollary}
		\begin{proof}
			Let $f_i$ be the family of continuous coordinate functionals from the previous theorem and $C$ the constant bounding their norms.
			Set $\alpha:= \frac 1{C+1}$.
			To see that $\alpha$ is as required note that since $(e_i)$ is a Schauder basis $e_i-e_j\neq 0$ for $i\neq j$.
			Therefore the vector can be normalized and it holds that
			\[ C+1 > \norm {f_i} \geq \abs{f_i\left(\frac{e_i-e_j}{\norm{e_i-e_j}}\right)} = \frac1{\norm{e_i-e_j}}. \]
			The assertion follows by taking reciprocals on both sides.
		\end{proof}

		We are now prepared to prove the main theorem of this section.
		\begin{theorem}[Representing Banach spaces]\label{resu:representing banach spaces}
			Let $\XX$ be an infinite dimensional separable Banach space and let $S:\omega^\omega\times \omega \to\omega$ be monotone and time-constructible such that for all $l\in\omega^\omega$ there exists an $l'\in\omega^\omega$ such that $S(l',\cdot)\geq l$.
			Then there exists an admissible, regular, complete representation $\xi$ of $\XX$ such that addition is computable in time $\bigo(l(n+1))$, the scalar multiplication in time $\bigo(l(n))$ and the norm in time $\bigo(T)$ for
			\[ T(l,n):=l(n+\lceil\lb(S(l,n+1)+1)+1\rceil)\cdot S(l,n+1) \]
			If the space $\XX$ allows for a Schauder basis, then $\xi$ can be chosen such that additionally there is a constant $C$ such that for all strictly increasing $l:\omega\to\omega$ and $n\in\omega$
			\[ S(l,n) \leq \size{\xi(K_l)}(n+C). \]
		\end{theorem}

		\begin{proof}
			Since the space $\XX$ is separable, there exists a fundamental system $(e_i)_{i\in\NN}$.
			If the space $\XX$ allows a Schauder basis, choose $(e_i)_{i\in\NN}$ as a Schauder basis. 
			Define $\xi$ as follows:
			A string function $\varphi$ is a name of $x\in \XX$ if and only if all of the following conditions hold:
			\begin{itemize}
				\item[(l)] $\varphi$ provides its \underline{l}ength: For all $n\in\NN$ it holds that $\length{\varphi(\sdzero^n)} = \flength{\varphi}(n)$.
			 	\item[(a)] $\varphi$ encodes linear combinations that \underline{a}pproximate $x$:
			 	For all $n\in\NN$ there exists a linear combination of the first $S(\flength\varphi,\length n)$ vectors $e_i$ that approximates $x$ with precision $\frac1{n+1}$ and whenever $m\in\NN$ is bigger than $(S(\flength\varphi,\length n +1)+1)(n+1)$ it holds that
				\[ \norm{\sum\nolimits_{i=0}^{S(\flength\varphi,\length n)}\frac{\varphi(\sdzero \langle i,n,m\rangle)}{m+1} e_i - x}\leq \frac2{n+1}. \]
				\item[(o)] $\varphi$ provides an \underline{o}racle for the norm:
				I.e. for all $n,m,N\in\NN$ and integers $z_0,\ldots,z_N\in\ZZ$
				\[ \abs{\norm{\sum\nolimits_{i=0}^N\frac{z_i}{m+1} e_i}-\frac{\varphi(\sdone\langle\langle z_0,\ldots,z_N\rangle,N,n,m\rangle)}{n+1}}\leq \frac1{n+1}. \]
			\end{itemize}
			
			This defines a representation due to the assumptions:
			To see that each element has a name note that due to the assumptions about the pairing function the requirements (l), (a) and (o) do not interfere with one another:
			The growth condition on $S$ and $e_i$ being a fundamental system guarantee that it is always possible to find a string function long enough to fulfill (a).
			This function can always be modified to a longer one fulfilling (o).
			Since $S$ is monotone, the first condition will still be fulfilled by this function.
			Afterwards its value on $\sdzero^n$ can be changed to have the length as length.
			To see that a name uniquely determines an element note that for a given precision requirement $2^{n}-1$ it suffices to choose $m$ in (a) bigger than $(S(\flength\varphi,\length n +1)+1)(n+1)$ to guarantee that there exist rational coefficients.
			These coefficients determine $x$ up to precision $2^{-n}$ and this works for any $n$.
			Therefore $x$ is uniquely determined from any of its names and $\xi$ is a representation.

			The straight-forward algorithms for the vector space operations are easily seen to run in the specified times.
			The norm can be computed as follows:
			Given a name $\varphi$ of some $x\in\XX$ as oracle and a precision requirement $n$, first compute $S(\flength\varphi,\length n+1)$.
			Since $S$ is time-constructible this is possible in time $\bigo(T)$.
			Next, query the given name $\varphi$ of some $x\in\XX$ to obtain integers $z_0,\ldots,z_{S(l,n+1)}$ such that
			\[ \norm{\sum\nolimits_{i=0}^{S(l,\length n)} \frac{z_i}{(S(\flength\varphi,\length n+1)+1)(n+1)}e_i-x}\leq \frac1{n+1}. \]
			To write each query, time linear in $\length n+\lceil\lb(S(\flength\varphi,\length n+1)+1)\rceil$ is needed.
			Since the length of each of the $S(\flength\varphi,\length n+1)$ answers is bounded by $\flength\varphi(\length n+\lceil\lb(S(\flength\varphi,\length n+1)+1)\rceil)$, the final query to the oracle for the metric can be asked in time $\bigo(T)$.
			This query obtains an approximation to the norm of this linear combination with precision $\frac 1{n+1}$.
			Writing the query can also be done in time $\bigo(T)$ and results in a valid return value for the norm.

			Let $(q_i)$ be a standard enumeration of the dyadic numbers.
			It is possible to translate back and forth between this representation and the relativized Cauchy representation with respect to the dense sequence
			\[ r_{\langle N,\langle i_1,\ldots,i_N\rangle\rangle} := \sum_{i=1}^N q_{j_i} e_i. \]
			The time needed to do this depends on $S$ but is bounded since $S$ is time-constructible.
			Therefore $\xi$ is admissible.
			To prove its completeness, use that the time-constructibility of $S$ implies continuity to prove that the images of a convergent sequence of names forms a Cauchy sequence.
			Since the space is complete, this sequence converges.
			Use the continuity of $S$ again to prove that the limit of the names is a name of the limit.

			To prove the lower bound on the size fix some $n$ and note that for all $i\in\NN$ with $\length i\leq S(l,\sdone^n)$ the elements $2^{-n}e_i$ have names of length $l$.
			Here the strict monotonicity guarantees that (o) can be fulfilled.
			By \Cref{resu:riezs property of schauder bases} there exists a constant $\alpha>0$ such that these elements are of pairwise distance more than $2^{-n}\alpha$.
			Set $C:= -\lceil\lb(\alpha)\rceil$, then the elements are of pairwise distance more than $2^{-n-C}$.
			Thus,
			\[ n\mapsto \begin{cases} 0 & \text{if }n<C \\ S(l,n-C) &\text{otherwise}\end{cases} \]
			is a spanning bound of $\xi(K_l)$ and therefore an upper bound of the metric entropy by \Cref{resu:spanning bounds}.
		\end{proof}

	\subsection{Computability issues}

		The last sections completely neglected considering computability issues of the metric by always providing its values via an oracle in each name.
		However, in practice computable metric spaces are very important.
		Under the assumption that the sequence of a compact computable metric space is uniformly dense it is straight forward to see that representation defined in the proof of \Cref{resu:representing compact spaces} is computably complete.

		\begin{corollary}
			Additionally to the assumptions of \Cref{resu:representing compact spaces} assume that $(r_i)_{i\in\NN}$ is a uniformly dense sequence such that $\M:=(\Mn,d,(r_i))$ is a computable metric space.
			Then there exists a representation $\xi$ as in the theorem  and which is additionally computably complete.
		\end{corollary}

		This result is unsatisfactory.
		It would be desirable to be able to construct the uniformly dense subsequence in a computable way from an dense sequence such that the corresponding discrete metric is computable.
		That is: It would be nice to have an effective version of \Cref{resu:uniformly dense}.
		However, the authors failed to produce such a result.
		As already mentioned at the beginning of the previous section the existence of a Schauder basis does not follow from the separability of a Banach space due to an example by Enflo.
		A version of this example provided by Bosserhoff proves that computability together with the existence of a Schauder basis does not guarantee the existence of a computable Schauder basis:
		\begin{theorem}[\cite{MR2534355}]
			There exists a computable Banach space that has a Schauder basis but does not have a computable Schauder basis.
		\end{theorem}
		If a computable Schauder basis exists, however, the representation constructed as in \Cref{resu:representing banach spaces} is computably complete.

	\section{Conclusion}
		
		In \Cref{resu:from complexity to metric entropy} there is still some space for improvements.
		It might be possible to remove the square in the bound of the metric entropy by further reducing the size bound from \Cref{resu:oracle restriction} by choosing the communication function $L$ cleverer:
		Instead of saving initial segments of length of the value of the runtime, one could explicitly follow the computation and only save those entries that are touched by the reading head.
		However, we consider the payoff to be small in comparison to the additional technical complications encountered in the proof.

		The question whether it is possible to computably construct a uniformly dense sequence from some additional information on the compact metric space, like a bound on the metric entropy or a spanning bound was left open.
		This is unsatisfactory, it would be desirable to either have a computable version of \Cref{resu:uniformly dense} or a counterexample as can be constructed for Banach spaces.
		The authors attempted only the former, but it seems like the covering property of a uniformly dense sequence is out of grasp of any computable procedure as it is notoriously difficult to verify it by a machine.

		In the solution theory of partial differential equations the finite element methods provide a powerful tool.
		The formulation of these results within computable analysis is an active field of research \cite{MR2275411}.
		Part of what makes finite element methods so valuable is that they can not only be applied to prove existence of solutions but also work well for implementations and provide fast algorithms for approximating the solutions.
		We consider a formulation of these methods within  real complexity theory our long time goal and hope that the content of this paper is a first step towards achieving this.

		Completeness may or may not lead to the right notion of bounded time admissibility.
		Even if it does, this does not end the quest for a good notion of polynomial-time admissibility.

		Results of the kind presented in this paper often allow refinements that replace polynomial time bounds by logarithmic space bounds \cite{Kawamura:2016:CTC:2933575.2935311}.
		However, space-bounded computation in presence of oracles is a tricky field.
		Classes of sub-linear space restrictions have been defined for real complexity theory \cite{MR3259646}.
		The results of this paper require classes with arbitrary space bounds and a considerably more elaborate machinery than introduced in this paper.

		We believe that relaxing the condition of being a second-order representation to regularity is an important step.
		We furthermore believe that arbitrary representations should be studied.
		One reason for doing so is that one of the most popular attempts to implement the ideas of real complexity theory, namely \texttt{iRRAM} \cite{iRRAM}, features polynomial-time evaluation of functions, while it seems impossible to efficiently find a modulus of uniform continuity.
		The minimality result for the standard representation of continuous functions on the unit interval proves that such a behavior can not be modeled with second-order representations.
		A representation such that it takes exponential time to compute a modulus of uniform continuity, but evaluation is possible in polynomial time, however, is not difficult to write down.
		Note, that this can not simply be achieved by dropping the condition (l) in the proof of \Cref{resu:representing banach spaces} as this leads to the loss of polynomial-time evaluation.

		As mentioned before many of the results from this paper were produced in an attempt to unify two approaches to real complexity theory.
		It should be mentioned that there exist still different frameworks, for instance the model of analog computation has recently made huge advancements \cite{DBLP:journals/corr/BournezGP16c}.
		It might be a good idea to look for interconnections.
\appendix{

\section{Relations to other work and examples}\label{sec:appl}

	Let $\M=(\Mn,d,(r_i))$ be a complete computable metric space.
	Recall from \Cref{ex:a family of compact sets} that it was possible to write down the sets $\xi_{\M}(K_l)$ of elements that have a short name in the Cauchy representation explicitly:
	\[ \xi_{\M}(K_l) = \bigcap_{n\in\NN}\bigcup_{i=0}^{2^{l(n)}-1} B_{2^{-n}}(r_i). \]
	From this, for non-decreasing $l$, the bound $\size{\xi_{\M}(K_l)}\leq l$ followed.
	If $\Mn$ is compact and the sequence $(r_i)$ is assumed to be uniformly dense in the sense of \Cref{def:uniformly dense} (actually it suffices that it has the spanning property (s)), then it is possible to also specify a lower bound: $\size{\xi_{\M}(K_l)}(n) \geq l(n)-1$ whenever $l(n) \leq \size{Y}(n-1)$.

	If we step up to the relativized Cauchy representation from \Cref{def:the relativized cauchy representation}, the same relations remain valid as long as the metric space is bounded and $l$ is strictly increasing with $l(0)\geq \lceil\lb(\diam(\Mn))\rceil$.
	If these additional assumptions are not made, then depending on the choice of the sequence $(r_i)$ there may only exist long oracles for the discrete metric.

	It is possible to give a similar description for the sets $\xi(K_l)$, where $\xi$ is the representation from the proof of \Cref{resu:representing compact spaces}:

	\begin{lemma}\label{resu:short name sets in compact spaces}
		Let $(Y,d)$ be an infinite compact metric space and $\ell$ and $S$ fulfill the assumptions of \Cref{resu:representing compact spaces} and let $\xi$ be the representation from the proof.
		Whenever $l$ is strictly increasing such that $l(0)\geq \lceil \lb(\diam(Y))\rceil$, then
		\[ \xi(K_l) = \bigcap_{n\in\NN}\bigcup_{i=0}^{2^{l(n)S(l,n)}-1} B_{2^{-n}}(r_i). \]
	\end{lemma}

	An upper bound on the metric entropy of $\xi(K_l)$ can be extracted from this like in \Cref{ex:a family of compact sets}: $\size{\xi(K_l)}(n)\leq \mu(n)+\lceil\lb(n+1)\rceil$.
	On the other hand, the spanning property (s) of a uniformly dense sequence (cf. \Cref{def:uniformly dense}) is equivalent being able to specify a lower bound on the set $\xi(K_l)$ whenever $l(n)\cdot S(l,n)\leq \size Y(n-1)$.
	Here, the additional assumption about $l$ is reasonable as it is necessary to guarantee that the candidate for the bound is not bigger than the compact space $Y$ itself.
	Without the requirement that the sequence is uniformly dense it is impossible to specify a lower bound on the metric entropy and in non-compact spaces there does not seem to be any sensible notion of uniform density.

	If a compact metric space is represented as in \Cref{resu:representing compact spaces}, then the sets $\xi(K_l)$ have an explicit description.
	The bounds on the size of these sets break down when compactness fails.
	Thus, for non-compact spaces the representation $\xi$ has to be changed in a way that $\xi(K_l)$ is a different set.
	In the case of a Banach space, families of compact sets that are candidates for the sets $\xi(K_l)$ have been investigated in approximation theory for a long time \cite{lorentz1966}. 
	
	\subsection{Full approximation sets}\label{sec:full approximation sets}

		This section compares the results of this paper to the results from approximation theory that inspired them in the first place \cite{lorentz1966}.
		For the convenience of the reader all relevant definitions and results from approximation theory are repeated here and then compared to our definitions.

		\begin{definition}[\cite{lorentz1966}]
			For a subset $A$ of a metric space $\Mn$ and $\varepsilon>0$ a real number, denote the minimal number of sets of diameter $2\varepsilon$ needed to cover $A$ by $N_{\varepsilon}(A)$.
			And define the \demph{entropy} of $A$ by
			\[ H_\varepsilon(A) := \log(N_\varepsilon(A)). \]
		\end{definition}
		Our notion of metric entropy $\size A$ from \Cref{def:metric entropy and size} uses closed balls instead of sets of diameter $2\varepsilon$.
		Since balls of radius $2^{-n}$ are sets of diameter $2\cdot2^{-n}$ it holds that $\lceil \lb(N_{2^{-n}}(A))\rceil\leq \size A(n)$.
		Any set of diameter $2\cdot2^{-n-1}$ is contained in the closed $2^{-n}$-ball around any of its elements.
		Therefore it also holds that $\size A(n)\leq \lceil \lb(N_{2^{-n-1}}(A))\rceil$.
		Note, that the $\log$ in the definition of $H_\varepsilon(A)$ denotes the natural logarithm, while our definition uses a power of two.
		It follows that
		\[ \lceil \lb(e) H_{2^{-n}}(A)\rceil \leq \size A(n)\leq\lceil\lb(e) H_{2^{-n-1}}(A)\rceil. \]

		\begin{definition}[\cite{lorentz1966}]
			Let $\XX$ be a Banach space, let $\Phi=(e_i)$ be a fundamental sequence, that is a sequence whose linear span is dense in $\XX$.
			Let $\Delta:= (\delta_i)$ be a non-increasing zero sequence of positive real numbers.
			Define the \demph{full approximation set} $A(\Delta,\Phi)$ by
			\[ A(\Delta,\Phi) := \big\{f\in\XX\mid\forall n\in\NN\exists \lambda_1\ldots\lambda_n\in\RR:\norm{\sum\nolimits_{i=1}^n \lambda_i e_i - f}\leq \delta_n\big\}. \]
		\end{definition}

		Note that a fundamental sequence is different from a fundamental system as introduced in \Cref{def:fundamental system} as it is not required that $\norm{e_i}=1$.
		Lorentz provides upper and lower bounds for the entropy of the approximation sets:
		\begin{theorem}[Theorem 2 from \cite{lorentz1966}]\label{resu:full approximation sets}
			Let $\XX$ be a Banach space and let $A(\Delta,\Phi)$ be a full approximation set.
			Define a sequence $N_i$ by $N_0:= 0$ and  $N_i:=\min\{k\mid \delta_k\leq 2^{-i}\}$ for $i\neq 0$.
			Set $\Delta N_i:= N_{i+1}-N_i$.
			Let $\varepsilon$ be from the interval $(0,4)$ and set $j:= \lceil2-\lb(\varepsilon)\rceil$.
			Then
			\[ \log(2)\sum_{i=1}^{j-3} N_i \leq H_{\varepsilon}(A(\Delta,\Phi)) \]
			and
			\[ H_\varepsilon(A(\Delta,\Phi)) \leq \log(2)\sum_{i=1}^jN_i + \sum_{i=0}^{j-1} N_i \log \frac{N_j}{\Delta N_i} + N_1\log\delta_0 + N_j\log 9. \]
		\end{theorem}
		Thus, these full approximation sets are compact and it is possible to specify both upper and lower bounds of their size.
		The representation $\xi$ from \Cref{resu:representing banach spaces} is deliberately chosen such that the sets $\xi(K_l)$ of elements that have names of length $l$ are almost full approximation sets:

		\begin{lemma}\label{resu:short name sets in banach spaces}
			Let $\XX$ be a Banach space, $\Phi=(e_i)$ a Schauder basis, $C$ the constant from \Cref{resu:uniform boundedness principle} and set $c:=\lceil\lb(C)\rceil$.
			Let $S$ be as in \Cref{resu:representing banach spaces}, let $\xi$ be the representation from the proof of \Cref{resu:representing banach spaces} and for a strictly increasing $l\in\omega^\omega$ such that $l(0)\geq c$ denote by $\Delta_l:= (\delta_{l,i})_i$ the sequence
			\[ \delta_{l,i} := \begin{cases}2^{l(0)-c} &\text{if } i<S(l,0)\\ 2^{-n}& \text{whenever }S(l,n)\leq i<S(l,n+1). \end{cases} \]
			And let $\Delta'_l$ be defined in the same way but replace $2^{l(0)-c}$ by $2^{l(3(S(l,1)+1)+1)S(l,0)}$.
			Then
			\[ A(\Delta_l,\Phi) \subseteq \xi(K_l) \subseteq A(\Delta_l',\Phi). \]
		\end{lemma}

		\begin{proof}
			First note that due to the monotonicity and the growth condition on $S$ from the statement of \cref{resu:representing banach spaces}, the function $S(l,\cdot)$ is non-decreasing and unbounded.
			Therefore $\delta_{l,i}$ is well-defined, and a non-increasing zero sequence as required.

			The first inclusion $\xi(K_l)\subseteq A(\Delta_l',\Phi)$ follows straightforwardly from comparing the definition of the representation $\xi$ from the proof of \Cref{resu:representing banach spaces} with the definition of full approximation sets and the sequence $\Delta_l$.
			
			To prove the second inclusion, namely $A(\Delta_l,\Phi) \subseteq \xi(K_l)$ assume that $x\in A(\Delta_l,\Phi)$.
			By the definition of the full approximation sets and the sequence elements $\delta_{l,i}$ for $i\geq S(l,0)$, there exists appropriate real numbers such that the first condition of (a) from the proof of \Cref{resu:representing banach spaces} can be fulfilled by a function of length $l$.
			Recall that we used $f_i$ to denote the coordinate functionals corresponding to the Schauder basis $e_i$ and that $c$ was chosen such that $\norm{f_i}\leq 2^c$.
			Therefore,
			\[ \abs{f_i(x)} \leq \norm{f_i} \norm x \leq 2^{l(0)} \]
			and since $l$ is strictly increasing there is enough space to encode approximations of $f_i(x)$ in a function of length $l$.
			That is: the second condition of (a) in the proof of \Cref{resu:representing banach spaces} can additionally be fulfilled.
			The conditions from (l) and (o) from the proof of \Cref{resu:representing banach spaces} can easily be fulfilled by modifying this name without changing its length.
			That the oracle for the norm does not need to increase the length is guaranteed by the normalization of the $(e_i)$.
		\end{proof}

		Let us apply Lorentz's \Cref{resu:full approximation sets} to the sets $\xi(K_l)$ and translate the result back to our definitions.
		The notations introduced in in \Cref{resu:full approximation sets} simplify considerably in this context.
		Indeed, if $\Delta=\Delta_l$ and $\varepsilon=2^{-n}$, then
		\[ N_i = S(l,i) \quad\text{and}\quad j=n+2. \]
		Also note that $\lb(e)\cdot \log(2) = 1$ and recall that
		\[ \lceil \lb(e) H_{2^{-n}}(A)\rceil \leq \size A(n)\leq\lceil\lb(e) H_{2^{-n-1}}(A)\rceil. \]

		\begin{theorem}\label{resu:full approximation name sets}
			Let $\XX$ be a Banach space with Schauder basis $(e_i)$, and $S$, $l$ and $c$ like in \Cref{resu:short name sets in banach spaces}.
			Then
			\begin{align*}
				\size{\xi(K_l)}(n) & \leq \sum_{i=1}^{n+3} S(l,i) + \sum_{i=0}^{n+2}S(l,i) \log\left(\frac{S(l,n+3)}{S(l,i+1)-S(l,i)}+2\right) + {} \\
				&\quad {}+ \lb(e)\cdot S(l,1)\cdot S(l,0)\cdot l(3(S(l,1)+1)+1)+{}\\&\quad{} + \log(7)\cdot S(l,n+3)
			\end{align*}
			and
			\[ \sum_{i=1}^{n-1} S(l,i) \leq \size{\xi(K_l)}(n) \]
		\end{theorem}
		The lower bound improves the lower bound given in \Cref{resu:representing banach spaces} considerably:
		Most prominently, the dependency on $c$ is completely removed (only the condition $l(0)\geq c$ remains).
		But also the bound from \Cref{resu:representing banach spaces} only mentions one of the summands from the sum in the lower bound form the theorem above (i.e. $S(l,n-c)$).

		The upper bounds are more difficult to compare.
		This is partly because we never stated the upper bound explicitly but it has to be obtained from combining \Cref{resu:representing banach spaces,resu:from complexity to metric entropy}.
		Also, these results do not state what the constants are.
		However, a superficial comparison indicates that the bound from \Cref{resu:full approximation name sets} is superior as it does not contain squares of $S(l,n+4)$ and states the constants explicitly.

	\subsection{Compact spaces}

		Let us start with a couple of examples of well-known representations that can be recovered using the constructions presented in \Cref{sec:construction of standard representations}.
		First for the construction for compact metric spaces from \Cref{resu:representing compact spaces}:
		\begin{example}[The unit interval]\label{ex:the unit interval}
			Let $[0,1]$ denote the computable metric space $([0,1],\abs\cdot,(q_i))$, where the sequence $(q_i)$ is given by
			\[ q_0:=0, \quad q_1:= 1 \quad\text{and}\quad q_i:=\frac{2\big(i-2^{\lfloor\lb(i-1)\rfloor}\big)-1}{2^{\lceil\lb(i)\rceil}}\text{ for }i\geq2. \]
			I.e.
			\[ (q_i) = \left(0,1,\frac12,\frac14,\frac34,\frac18,\frac38,\frac58,\frac78,\frac1{16},\frac3{16}\ldots\right). \]
			It is easy to check that $\size{[0,1]}(n) =\max\{n-2,0\}$ and that the sequence $(q_i)$ is uniformly dense in the sense of \Cref{def:uniformly dense}.
			Actually, $(q_i)$ fulfills even stronger conditions than (c) and (s) from this definition:
			(c) can be replaced by \lq The closed $2^{-n}$-balls around the first first $2^{\size{[0,1]}(n)+1}$ elements cover $[0,1]$\rq\ and (s) can be improved to \lq For all $n$ the first $2^{n+1}-1$ elements are of pairwise distance more than $2^{-n}$\rq.
			Let $\xi_{[0,1]}$ be the representation constructed in the proof of \Cref{resu:representing compact spaces} with the choices
			\[ \ell(n) := n\quad\text{and}\quad S(l,n):=1. \]
			These do not fulfill the assumptions of the theorem but due to the improved uniform density of the sequence $(q_i)$ the construction still works and returns a representation which is equivalent to the range restriction of the representation of the reals from \Cref{def:standard representation of the reals}.

			We could also have used \Cref{resu:representing compact spaces} directly by choosing the functions as $\ell(n) := \size{[0,1]}(n) + \lceil\lb(n+1)\rceil$ and $S(l,n) := 1$.
			This leads to an polynomial-time equivalent representation.
		\end{example}

		\begin{example}[Separable metric spaces]
			If $\M=(\Mn,d,(r_i))$ is a separable metric space with a uniformly dense sequence, using $(r_i)$ as uniformly dense sequence and setting $S(l,n):=1$ and $\ell(n) := \size{\Mn}(n)+\lceil\lb(n+1)\rceil$ produces the relativized Cauchy representation $\xi^r_{\M}$ from \Cref{def:the relativized cauchy representation}.
			This still works if the sequence is only dense, but in this case the lower bound for the size of the sets $\xi^r_{\M}(K_l)$ may fail.
		\end{example}

		As a last example of a compact space consider a space of Lipschitz functions \cite{MR1952428}.
		Since the procedure is similar to what we do in the next section we skip the details:

		\begin{example}\label{ex:continuous functions}
			The standard enumeration of piecewise linear Lipschitz one functions with dyadic breakpoints is uniformly dense in the set of Lipschitz one functions that map $0$ to itself.
			Using the construction from \Cref{resu:representing compact spaces} for this sequence and setting $S(l,n) := 2^{\max\{l(n),n\}}$ and $\ell(n):=n$ or $S(l,n):=2^{l(n)}+\lceil \lb(n+1)\rceil$ and $\ell(n):=1$ produces the range restriction of the standard representation of the continuous functions as introduced in \cite{Kawamura:2012:CTO:2189778.2189780} (see also \Cref{def:the standard representation of continuous functions}).
			I.e. the two representations are polynomial time equivalent.
		\end{example}
		
	\subsection{Continuous functions on the unit interval}
		
		This chapter demonstrates how to reconstruct the standard representation of the continuous functions on the unit interval (as introduced in \cite{Kawamura:2012:CTO:2189778.2189780}) using the construction from the proof of \Cref{resu:representing banach spaces}.
		This chapter relies on the framework of second-order representations introduced by Kawamura and Cook (see \Cref{sec:regularity} for a brief description).
		For the convenience of the reader the relevant definitions from the above source are repeated.
		In the following the space $C([0,1])$ of real-valued continuous functions on the unit interval is always considered a Banach space equipped with the supremum norm denoted by $\norm\cdot_\infty$.

		A central notion for the standard representation is the modulus of continuity.
		\begin{definition}[\cite{Kawamura:2012:CTO:2189778.2189780}]\label{def:moduli of continuity}
			A non-decreasing function $\mu:\omega\to\omega$ is called a \demph{modulus of continuity} of $f\in C([0,1])$ if for all $x,y\in[0,1]$ and $n\in\omega$ it holds that
			\[ \abs{x-y}\leq 2^{-\mu(n)} \quad\Rightarrow\quad |f(x)-f(y)| \leq 2^{-n} \]
		\end{definition}
		To be exact, this notion of a modulus should be called modulus of uniform continuity.
		Since functions on compact sets are uniformly continuous, any function from $C([0,1])$ has a modulus of continuity.

		\begin{definition}[\cite{Kawamura:2012:CTO:2189778.2189780}]
			\label{def:the standard representation of continuous functions}
			Make $C([0,1])$ a represented space by equipping it with the \demph{standard second-order representation $\delta_{\square}$} defined as follows:
			A length-monotone string function $\varphi$ is a $\delta_{\square}$-name of $f$ if and only if $\varphi = \langle\mu,\psi\rangle$ for some $\mu,\psi\in\reg$ such that both of the following are fulfilled:
			\begin{itemize}
			 	\item[(m)] $n\mapsto \length{\mu(\sdone^n)}$ a \underline{m}odulus of continuity of $f$.
			 	\item[(d)] $\varphi$ provides approximations of the values of $f$ on \underline{d}yadic arguments.
			 	I.e. for all $n,m\in\omega$ and $r\in\NN$ such that $r\leq 2^m$ it holds that $\varphi(\langle \sdzero^n,r,\sdone\sdzero^m\rangle)=\langle q,\sdone \sdzero^k\rangle$ for some $q\in\NN$ and $k\in\omega$ and
			 	\[ \abs{f\left(r\cdot 2^{-m}\right) - q\cdot2^{-k}}\leq 2^{-n}. \]
			\end{itemize}
		\end{definition}
		Note that this definition differs subtly from the definitions used earlier in this paper:
		We chose to use integers in binary instead of unary, but have already remarked that this does not make a difference.
		Additionally we have always chosen what is called $k$ in the above definition to be equal to what is called $n$ in the above definition.
		This is impossible if one wants to work within the realm of second-order representations: it must be possible to find length-monotone names, so it has to be possible to choose the return values arbitrary big.
		This does not make a difference up to polynomial-time translatability, as we can always just round the dyadic number to have denominator $2^{-n}$.
		In the other direction the fraction can be extended for the return value to be long enough as soon as we have access to an upper bound of the length function (which is the case if the representation is regular).
		Finally, considering names as pairs of a modulus of continuity and a string function encoding discrete information is up to polynomial-time equivalence the same as just using the function encoding the discrete information and requiring its length to be a modulus of continuity.

		We need the following important result about the representation $\delta_{\square}$:
		\begin{theorem}[Lemma 4.9 from \cite{Kawamura:2012:CTO:2189778.2189780}]\label{resu:minimality of the standard representation}
			For a second-order representation $\delta$ of $C([0,1])$ the following are equivalent:
			\begin{itemize}
				\item The evaluation operator $\eval:C([0,1])\times [0,1] \to \RR, (f,x)\mapsto f(x)$ is polynomial time computable. 
				\item $\delta$ can be translated to $\delta_{\square}$ in polynomial-time.
			\end{itemize}
		\end{theorem}

		Let us introduce the Schauder basis to use in the construction from \Cref{resu:representing banach spaces}.

		\noindent\begin{minipage}{.5\textwidth}
		\begin{definition}[\cite{Faber1910}]\label{def:faber-schauder system}
			Let $(q_i)$ be the uniformly dense sequence in the unit interval from \Cref{ex:the unit interval} and choose the convention $\lceil\lb(0)\rceil=0$.
			The \demph{Faber-Schauder system} is the sequence $(e_i)$ of functions $e_i\in C([0,1])$ defined by
			\[ e_i(x) := \max\big\{1-2^{\lceil\lb(i)\rceil}\cdot\abs{x-q_i},0\big\} \]
		\end{definition}
		It is well-known that the Faber-Schauder system is a Schauder basis.
		Since it is instructive

		\vspace{.175cm}
		\end{minipage}
		\begin{minipage}{.4\textwidth}
			\centering
			\begin{tikzpicture}
				\draw[->] (0,-.1) -- (0,2.5);
				\draw[->] (-.1,0) -- (2.5,0);
				\draw[thick, color = red7] (0,0) -- (1,0) -- (1.25,2) -- (1.5,0)--(2,0);
				\draw[thick, color = red6] (0,0) -- (.5,0) -- (.75,2) -- (1,0)--(2,0);
				\draw[thick, color = red5] (0,0) -- (.25,2) -- (.5,0)--(2,0);
				\draw[thick, color = red4] (0,0) -- (1,0) -- (1.5,2)--(2,0);
				\draw[thick, color = red3] (0,0) -- (.5,2) -- (1,0)--(2,0);
				\draw[thick, color = red2] (0,0) -- (1,2) -- (2,0);
				\draw[thick, color = red1] (0,0) -- (2,2);
				\draw[thick, color = red0] (2,0) -- (0,2);
				\node at (2.5,.25) {\textcolor{red0}{$e_0$}};
				\node at (2.5,.5) {\textcolor{red1}{$e_1$}};
				\node at (2.5,.75) {\textcolor{red2}{$e_2$}};
				\node at (2.5,1) {\textcolor{red3}{$e_3$}};
				\node at (2.5,1.25) {\textcolor{red4}{$e_4$}};
				\node at (2.5,1.5) {\textcolor{red5}{$e_5$}};
				\node at (2.5,1.75) {\textcolor{red6}{$\vdots$}};

			\end{tikzpicture}
			\captionof{figure}{The functions $e_i$}
		\end{minipage}
		for the following, we repeat the proof of this fact:
		
		\begin{lemma}
			The Faber-Schauder system is a Schauder basis of the Banach space $C([0,1])$.
		\end{lemma}

		\begin{proof}
			First assume that $f$ is a function such that there are real numbers $\lambda_i$ that fulfill
			\[ f = \lim_{n\to \infty}\sum \lambda_i\cdot e_i. \]
			On one hand, the sequences $(e_i)$ and $(q_i)$ fulfill that $e_i(q_j) \neq 0$ implies $i\leq j$.
			This implies that
			\begin{equation}\label{eq:e}
				f(q_j) = \sum_{i=0}^{j} \lambda_i e_i(q_j)\tag{evl}
			\end{equation}
			\
			On the other hand, all but $q_0$ and $q_1$ are included in the interior of the interval.
			Denote the dyadic numbers that arise by rounding the last digit of $q_j$ up resp.\ down by $q_{j_+}$ resp.\ $q_{j_-}$.
			Since for all numbers $q_i$ with $i\geq j$ it holds that $e_i(q_j)=0$, the linear combination of the first $j$ elements has to reproduce the values of $f$ on the $q_i$ with $i\leq j$ exactly.
			Therefore it holds that 
			\begin{equation}\tag{lam}\label{eq:rec}
				\lambda_j =  f(q_j)-\frac{f(q_{j_-})+f(q_{j_+})}2.
			\end{equation}
			It follows that the sequence $(\lambda_i)$ is uniquely determined by the values of $f$ on dyadic numbers.
			For an arbitrary function $f$, the sequence defined above can be seen to always be such that
			\[ f = \lim_{n\to\infty} \sum_{i=0}^n\lambda_i e_i \]
			Therefore, $(e_i)$ is a Schauder basis.
		\end{proof}

		Note that $e_i(q_j)\neq 0$ also implies $\abs{q_i-q_j}<2^{-\lceil\lb(j)\rceil}$ and therefore that in the sum that defining $\lambda_{j+1}$ in the previous proof has at most $2\cdot\lceil\lb(j)\rceil$ summands are not zero.
		Furthermore, it can quickly be checked which are the ones that are not zero by evaluating $\abs{q_i-q_j}$.
		The last piece we need to see that the representation of $C([0,1])$ is polynomial-time equivalent to one constructed in the proof of \Cref{resu:representing banach spaces} is that $\lambda_i\cdot e_i$ has the function $n\mapsto n+\lceil\lb(i)\rceil+\lceil\lb(\abs{\lambda_i})\rceil$ as modulus of continuity.

		\begin{theorem}[Reconstruction of the standard representation]\label{resu:reconstruction of the standard representation}
			The standard second-order representation of the continuous functions on the unit interval is polynomial-time equivalent to the representation $\xi$ that is produced by the proof of \Cref{resu:representing banach spaces}, when choosing $(e_i)$ as the Faber-Schauder system and $S(l,n):=2^{\max\{l(n),n\}}$.
		\end{theorem}

		\begin{proof}
			To see that the representation $\xi$ can be translated to the representation $\delta_{\square}$ first note that $\xi$ as a regular representation is polynomial-time equivalent to a second-order representation by \Cref{resu:regularity vs. second-order}.
			Therefore using the minimality of $\delta_{\square}$ from \Cref{resu:minimality of the standard representation} it suffices to specify an polynomial-time algorithm for evaluation.

			This can be done using \Cref{eq:e}.
			Given a pair $\langle\varphi,\psi\rangle$ of a $\xi$-name $\varphi$ of a function $f$ and an $\RR$-name $\psi$ of a number $x$ as oracle and an integer $n$ as input a machine computing the evaluation may proceed as follows.
			It obtains values of $\length{\varphi}$ from $\varphi$ using condition (l) of the definition of the representation $\xi$: $\varphi(\sdzero^n)\geq \length{\varphi}(n)$.
			Due to the choice $S(l,n) = 2^{\max\{l(n),n\}}$, from condition (a) of the definition of $\xi$ it follows that a linear combination of the first $2^{\max\{\length{\varphi}(\length n),\length n\}}$ elements of $(e_i)$ approximates the function $f$ with precision $\frac1{n+1}$ in supremum norm.
			Note that if $\mu$ and $\nu$ are moduli of continuity of $f$ and $g$, then a modulus of continuity of $f+g$ is given by $n\mapsto \max\{\mu(n+1),\nu(n+1)\}$.
			The sum approximating $f$ has at most
			\[ 2\lb\big(2^{\max\{\flength\varphi(\length n),\length n\}}\big) = 2\max\{\flength\varphi(\length n),\length n\} \]
			summands.
			Furthermore, the machine can obtain rational numbers $\lambda_i$ that can be used as coefficients and $\lceil\lb(\abs{\lambda_i})\rceil$ is bounded in absolute value by $\flength{\varphi}(\flength\varphi(\length n+1) +\lceil\lb(n+1)\rceil)$ (compare to the part of the proof of \Cref{resu:representing banach spaces} that provides that a name determines an element).
			Therefore, a modulus of continuity of the sum can be specified as
			\[ \mu_n(m) := m+3\max\{\flength{\varphi}(\length n),\length n\} + \flength{\varphi}(\flength\varphi(\length n+1) +\lceil\lb(n+1)\rceil). \]
			Use the name $\psi$ of the input $x$ to obtain a rational approximation to $x$ with precision requirement $2^{-\mu_n(\length{n}+1)+1}$ and round this to a dyadic approximation of at least half the quality.
			Note that the sequence $(q_i)$ lists all dyadic numbers, therefore the approximation is of the form $q_j$ for some $j\in\NN$.
			The index $j$ can easily be found from an encoding of $q_j$.
			The machine now evaluates the sum
			\[ \sum_{i=0}^{2^{\max\{\flength\varphi(\length n),n\}}} \lambda_i e_i(q_j) \]
			By first listing the at most $\flength\varphi(\length n)$ values of $i$ such that the corresponding summand is not zero, then getting the $\lambda_i$ from $\varphi$ and evaluating $e_i(q_j)$.
			All of this can be done in polynomial time.

			To see that a $\xi$-name can be written down using the information encoded in a $\delta_{\square}$-name first note that using a modulus of continuity as $\flength{\varphi}$ fulfills the first requirement of condition (a) from the proof of \Cref{resu:representing banach spaces} and that (l) and (o) are unproblematic.
			To fulfill the second requirement of (a) valid coefficients $\lambda_i$ have to be specified.
			Those can be found from a $\delta_{\square}$-name by using the expression of $\lambda_j$ in terms of the values of $f$ on dyadic numbers from \eqref{eq:rec} and that $j_+$ and $j_-$ can be obtained from $j$.
		\end{proof}

	\subsection{Spaces of integrable functions}

		This section considers the spaces $\Lp([0,1])$.
		In this section it is always assumed that $1\leq p<\infty$.
		Recall that the space $\mathcal L^p(\RR)$ consists of functions $f:\RR\to \RR$ such that
		\[ \int_0^1 \abs{f}^p \dd\lambda < \infty \]
		and that $\norm f_p$ is defined to be the $p$-th root of the above quantity.
		Since functions that are equal to zero almost everywhere fulfill $\norm f_p=0$, this does not define a norm on $\mathcal L^p(\RR)$.
		It does define a norm on the space $\Lp(\RR)$ of  equivalence classes of such functions up to equality almost everywhere.
		Let $\Lp([0,1])$ denote the subspace of equivalence classes of functions whose support is contained in the unit interval.
		The mapping $\norm\cdot_p$ is well-defined on the equivalence classes (i.e. does not change under change of the representative) and makes this space a Banach space.
		In the following we sometimes refer to the elements of $\Lp([0,1])$ as functions, as this does rarely lead to confusion.

		\cite{arXiv:1612.06419} introduces representations for the spaces $\Lp([0,1])$ by imitating the construction of the standard representation of continuous functions:
		The point evaluations are replaced by integrals over dyadic intervals and the modulus of continuity by an integral modulus.
		For the convenience of the reader we repeat the necessary definitions and relevant results here.

		If $f \in \Lp([0,1])$ is the equivalence class of some function $g$, then let $\tau_h f\in\Lp(\RR)$ denote the equivalence class of the function $x\mapsto g(x+h)$.
		\begin{definition}[\cite{arXiv:1612.06419}]
			A function $\mu:\omega\to\omega$ is called an $\Lp$-modulus of $f\in\Lp([0,1])$ if for all $h>0$ and $n\in\NN$
			\[ \abs{h}\leq 2^{-\mu(n)} \quad \Rightarrow\quad \norm{f-\tau_h f}_p\leq 2^{-n}. \]
		\end{definition}

		The following definition is functionally identical to the definition from \cite{arXiv:1612.06419}.
		The formal differences are due to the use of an encoding of dyadic numbers in that source which this paper avoided.
		\begin{definition}[\cite{arXiv:1612.06419}]
			Define a second-order representation $\xi_p$ of $\Lp([0,1])$ as follows:
			A length monotone string function $\varphi$ be a name of $f\in\Lp([0,1])$ if and only if for any $k,l,m,n \in\NN$ it holds that $\varphi(\langle k,l,m,n\rangle)=\langle q,\sdzero^j \rangle$ for some $j\in\NN$ and $q\in\ZZ$ and
			\[ \abs{\int_{k\cdot 2^{-m}}^{l\cdot 2^{-m}} f \dd\lambda - q\cdot 2^{-j}}< 2^{-n} \]
			and $\flength \varphi$ is an $\Lp$-modulus of $f$.
		\end{definition}
		Note that in contrast to the most other examples the inequality in the definition is chosen to be strict.
		This convention was chosen in \cite{arXiv:1612.06419} to guarantee that the resulting representation is an open mapping.
		Whether or not the inequality is chosen to be strict does not make a difference up to polynomial-time equivalence.

		In \cite{arXiv:1612.06419} it is proven that this representation is computably equivalent to the Cauchy representation of $\Lp([0,1])$ if dyadic step functions are chosen as approximating sequence.
		To reconstruct this representation using \cref{resu:representing banach spaces}, it would therefore suggest  itself to use multiples of characteristic functions of dyadic intervals as a Schauder basis.
		However, these systems are typically not Schauder bases as it is easy to write a characteristic function as an infinite sum of characteristic functions of subintervals.
		We use the Haar system, which originates from \cite{MR1511592}, instead.
		However, since we reqiure the elements of a Schauder basis to have norm one we need different scaling for different values of $p$. \vspace{.15cm}

		\noindent\begin{minipage}{.6\textwidth}
		\begin{definition}[\cite{MR1511592}]
			Let $(q_i)$ be the uniformly dense sequence in the unit interval from \Cref{ex:the unit interval}.
			The \demph{$p$-normalized Haar system} is the sequence $(f_{i,p})$ of functions $f_{i,p}\in L^p([0,1])$ by $f_0 := 1$ and for $i=2,\ldots$ by
			\[ f_{i-1,p}(x) := \begin{cases}2^{\frac{\lceil\lb(i)\rceil-1}p} &\text{if } x\in [q_i-2^{-\lceil\lb(i)\rceil},q_i) \\ -2^{\frac{\lceil\lb(i)\rceil-1}p} &\text{if } x\in [q_i,q_i+2^{-\lceil\lb(i)\rceil}] \\ 0 &\text{otherwise.} \end{cases} \]
		\end{definition}
		The functions $f_{i,1}$ coincide with the almost everywhere defined derivative of the functions $e_{i+1}$ from the Faber-Schauder system from \Cref{def:faber-schauder system}.
		It is well-known that the Haar system is a Schauder ba-

		\vspace{.1cm}
		\end{minipage}\hspace{.5cm}
		\begin{minipage}{.25\textwidth}
			\hspace{-.25cm}
			\vspace{-.2cm}
			\begin{tikzpicture}
				\draw[->] (0,-2) -- (0,2);
				\draw[->] (-.1,0) -- (2.5,0);
				\draw[thick, color = red7,dotted] (1.75,0) -- (1.75,2);
				\draw[thick, color = red7] (1.5,2) -- (1.75,2);
				\draw[thick, color = red7,dotted] (1.75,2) -- (1.75,-2);
				\draw[thick, color = red7] (1.75,-2) -- (2,-2);
				\draw[thick, color = red6,dotted] (1,0) -- (1,2);
				\draw[thick, color = red6] (1,2) -- (1.25,2);
				\draw[thick, color = red6,dotted] (1.25,2) -- (1.25,-2);
				\draw[thick, color = red6] (1.25,-2) -- (1.5,-2);
				\draw[thick, color = red6,dotted] (1.5,-2) -- (1.5,0);
				\draw[thick, color = red5,dotted] (.5,0) -- (.5,2);
				\draw[thick, color = red5] (.5,2) -- (.75,2);
				\draw[thick, color = red5,dotted] (.75,2) -- (.75,-2);
				\draw[thick, color = red5] (.75,-2) -- (1,-2);
				\draw[thick, color = red5,dotted] (1,-2) -- (1,0);
				\draw[thick, color = red4] (0,2) -- (.25,2);
				\draw[thick, color = red4,dotted] (.25,2) -- (.25,-2);
				\draw[thick, color = red4] (.25,-2) -- (.5,-2);
				\draw[thick, color = red4,dotted] (.5,-2) -- (.5,0);
				\draw[thick, color = red4] (.5,0) -- (2,0);
				\draw[thick, color = red3] (0,0) -- (1,0);
				\draw[thick, color = red3,dotted] (1,0) -- (1,1);
				\draw[thick, color = red3] (1,1) -- (1.5,1);
				\draw[thick, color = red3,dotted] (1.5,1) -- (1.5,-1);
				\draw[thick, color = red3] (1.5,-1) -- (2,-1);
				\draw[thick, color = red2] (0,1) -- (.5,1);
				\draw[thick, color = red2,dotted] (.5,1) -- (.5,-1);
				\draw[thick, color = red2] (.5,-1) -- (1,-1);
				\draw[thick, color = red2,dotted] (1,-1) -- (1,0);
				\draw[thick, color = red2] (1,0) -- (2,0);
				\draw[thick, color = red1] (0,.5) -- (1,.5);
				\draw[thick, color = red1,dotted] (1,.5) -- (1,-.5);
				\draw[thick, color = red1] (1,-.5) -- (2,-.5);
				\node at (2.5,.5) {\textcolor{red1}{$f_1$}};
				\node at (2.5,.75) {\textcolor{red2}{$f_2$}};
				\node at (2.5,1) {\textcolor{red3}{$f_3$}};
				\node at (2.5,1.25) {\textcolor{red4}{$f_4$}};
				\node at (2.5,1.5) {\textcolor{red5}{$f_5$}};
				\node at (2.5,1.75) {\textcolor{red6}{$\vdots$}};
				\node at (-.1,.5) {1};
				\node at (-.1,1) {2};
				\node at (-.1,2) {4};
			\end{tikzpicture}
			
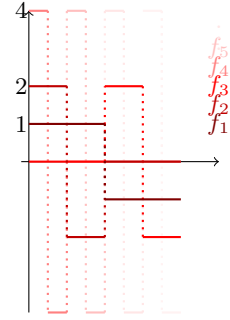
\captionof{figure}{The\newline functions $f_{i,1}$}
		\end{minipage}
		sis.
		Since it is instructive for the following, we repeat the proof of this fact:

		\begin{lemma}
			The $p$-normalized Haar system is a Schauder basis of the Banach space $\Lp([0,1])$.
		\end{lemma}

		\begin{proof}
			First assume that $f\in\Lp([0,1])$ is a function such that there are real numbers $\lambda_{i,p}$ such that with respect to the $\Lp$-norm
			\[ f = \lim_{N\to \infty}\sum_{i=0}^N \lambda_{i,p}\cdot f_{i,p}. \]
			Note that the sequence $(f_{i,p})$ fulfills that for all $i,j$
			\[ \int_{q_j-2^{-\lceil\lb(j)\rceil}}^{q_{j}} f_{i,p}\dd\lambda \neq 0 \quad \Rightarrow \quad i\leq j, \]
			and that the above integral takes the value $2^{-\frac{1}p}$ for $i=j$.
			On one hand it follows that
			\begin{equation}\tag{int}\label{eq:i}
				\int_{q_j-2^{-\lceil\lb(j)\rceil}}^{q_{j}} f\dd\lambda = \sum_{i=0}^{j-1} \lambda_{i,p} \int_{q_j-2^{-\lceil\lb(j)\rceil}}^{q_{j}} f_{i,p}\dd\lambda.
			\end{equation}
			On the other hand let $q_{j_+}$ be the dyadic number that arise by rounding the last digit of $q_j$ up.
			Since the integrals of $f_{i,p}$ over $[q_j-2^{-\lceil\lb(j)\rceil},q_j]$ vanish whenever $i\geq j$, these integrals have to be assumed by the finite linear combination broken of at $j$.
			Furthermore note that the integral over the finite linear combination changes linearly between the $q_i$ with $i\leq j$.
			Therefore $\lambda_{j,p}$ can be expressed as:
			\begin{gather*}\tag{lamp}\label{eq:lamp}
				\lambda_{j,p} = 2^{\frac 1p}\left(\int_{q_{j}-2^{-\lceil \lb(j)\rceil+1}}^{q_{j}} f\dd\lambda- \frac12\int_{q_{j_+}-2^{-\lceil \lb(j_+)\rceil}+2}^{q_{j_+}}f\dd\lambda\right).
			\end{gather*}
			It follows that the sequence is unique.
			On the other hand for an arbitrary function $f\in\Lp([0,1])$, the sequence defined above is always such that
			\[ f = \lim_{n\to\infty} \sum_{i=0}^n\lambda_{i,p} f_{i,p} \]
			in $\Lp([0,1])$.
			Furthermore, $\norm{f_{i,p}}_p= 1$.
			Therefore, $(f_{i,p})$ is a Schauder basis in $\Lp([0,1])$.
		\end{proof}

		Note that $\int_{q_j-2^{-\lceil\lb(j)\rceil}}^{q_j} f_i \dd\lambda\neq 0$ also implies that $\abs{q_i-q_j}\leq 2^{-\lceil\lb(j)\rceil}$ and therefore the sum in \Cref{eq:i} does at most have $2\cdot \lceil\lb(j)\rceil$ summands that do not vanish.

		\begin{theorem}[Reconstruction of $\xip$]\label{resu:reconstruction of xip}
			Whenever $p$ is polynomial-time computable, the second-order representation $\xip$ of the space $\Lp([0,1])$ is poly\-nomial-time equivalent to the representation $\xi$ that is produced by the proof of \Cref{resu:representing banach spaces}, when choosing the Haar system as Schauder basis and $S(l,n):=2^{\max\{l(n),n\}}$.
		\end{theorem}

		\begin{proof}[That $\xi$ can be translated to $\xip$]
			It is necessary to compute the integrals over dyadic intervals.
			Given a $\xi$-name $\varphi$ of a function $f\in \Lp([0,1])$ and integers $k,l,m,n$ as input a machine computing the corresponding integral may proceed as follows.
			It obtains values of $\length{\varphi}$ from $\varphi$ using condition (l) of the definition of the representation $\xi$: $\length{\varphi(\sdzero^n)}\geq \flength{\varphi}(n)$.
			Due to the choice $S(l,n) = 2^{\max\{l(n),n\}}$, from condition (a) of the definition of $\xi$ it follows that a linear combination of the first $2^{\max\{\length{\varphi}(n),n\}}$ elements of $(f_{i,p})$ approximates the function $f$ with precision $\frac1{n+1}$ in $\Lp$-norm.
			Note that the interval $[k\cdot 2^{-m},l\cdot 2^{-m}]$ can be written as a union of $2\cdot m$ intervals of the form $[q_j-2^{-\lceil\lb(j)\rceil},q_j]$.
			The integrals over the latter can be read from the coefficients $\lambda_{j,p}$ via \Cref{eq:i}, where it is important that at most a logarithmic number of integrals over the $f_{i,p}$ is non-zero.
			Furthermore, from the polynomial-time computability of $p$ it follows that the integrals over $f_{i,p}$ are computable in polynomial time.

			It is also necessary to obtain the value of an $\Lp$-modulus.
			an $\Lp$-modulus $\mu_{i,p}$ of $f_{i,p}$ can be specified as $\mu_{i,p}(n) = \lceil p\rceil(n+1) + \lceil\lb(i)\rceil$.
			Also an $\Lp$-modulus of the sum of two functions is given by $n\mapsto \max\{\mu(n+1),\nu(n+1)\}$ if $\mu$ and $\nu$ are $\Lp$-moduli of the summands.
			The combination of these observation results in the bound
			\[ \lceil p\rceil(n+\lceil\lb(S(\flength{\varphi},\length n))\rceil+1) + S(\flength\varphi,\length n + \lceil\lb(S(\flength{\varphi},\length n))\rceil) \]
			of an $\Lp$-modulus of $f$ in $\length n-2$.
			The machine can obtain this bound from $\varphi$ and pad the return value accordingly.
		\end{proof}

		The other of the proof needs results from \cite{arXiv:1612.06419} and additionally the following technical lemma:
		\begin{lemma}\label{resu:stepfunctions to haar basis}
			For any $i,j$ there are rational constants $\lambda_{k,p}$ such that the characteristic function of the interval $[q_i,q_j]$ can be written as
			\[ \chi_{[q_i,q_j]} = \sum_{k=0}^{\max\{i,j\}} \lambda_{k,p} f_{k,p}. \]
		\end{lemma}
		\begin{proof}
			Proceed by induction over $N:=\max\{i,j\}$.
			Since within $\Lp([0,1])$ it holds that $\chi_{[q_0,q_0]}\equiv 0 = 0f_0$ the base case holds.
			Now assume the statement is true for all $i',j'$ such that $\max\{i',j'\}=N$.
			Let $i,j$ be a pair such that $\max\{i,j\}=N+1$ if both $i$ and $j$ equal to $N+1$, then $\chi_{[q_{i},q_{j}]}\equiv 0 = 0f_0$, therefore assume without loss of generality that $i=N+1$ and $j\leq N$.
			Note that either $q_{i+1} = q_i+2^{-\lceil\lb(i)\rceil+1}$ (if $N$ is not a power of two) or otherwise $q_{i+1}= \frac1{2(i-1)} = q_{(i-1)/2}/2$
			First handle the former case: Note that both $q_i+2^{-\lceil\lb(i)\rceil} = q_l$ as well as $q_{i+1}+2^{-\lceil\lb(i)\rceil} = q_m$, where $l,m\leq i$.
			By the induction hypothesis both $\chi_{[q_i,q_l]}$ and $\chi_{[q_l,q_m]}$ can be expanded in a sum $f_k$s with indices smaller than $i$.
			Furthermore,
			\[ \chi_{[q_{l},q_{i+1}]} =  \frac{\chi_{[q_l,q_m]}+2^{\frac{1-\lceil\lb(i)\rceil}p}f_{i+1,p}}2 \]
			and $\chi_{[q_{i+1},q_j]} = \chi_{[q_i,q_j]} - \chi_{[q_i,q_l]}- \chi_{[q_l,q_{i+1}]}$.
			This provides the expansion of $\chi_{[q_{i+1},q_j]}$.
			If $N$ is a power of two, expand $\chi_{[q_{(i-1)/2},q_j]}$ instead of $\chi_{[q_i,q_j]}$ and apply the same argument.
		\end{proof}
		Our proof that the representation $\xip$ can be translated to the representation constructed in \Cref{resu:representing banach spaces} heavily relies two lemmas proven in \cite{arXiv:1612.06419}.
		For a fixed function $f\in\Lp([0,1])$ consider the smoothing operators $A_m:\Lp([0,1])\to C(\RR)$ defined by
		\[ A_m(f)(x):= 2^{dm} \int_{x-2^{-m}}^{x+2^{-m}} f\dd\lambda. \]
		\begin{lemma}[Continuity, \cite{arXiv:1612.06419}]\label{resu:continuitity of approximations}
			Whenever $\mu$ is an $\Lp$-modulus of $f\in\Lp([0,1])$, the function $n\mapsto \mu(n+m)$ is a modulus of continuity of $A_m(f)$.
		\end{lemma}
		\begin{lemma}[Approximation, \cite{arXiv:1612.06419}]\label{resu:approximation}
			Let $\mu$ be an $\Lp$-modulus of $f$.
			It holds that
			\[ \norm{f-A_{\mu(n)}(f)}_p < 2^{-n}. \]
		\end{lemma}
		Now we can construct the second translation for and prove its correctness \Cref{resu:reconstruction of xip}.
		\begin{proof}[That $\xip$ can be translated to $\xi$.]
			An oracle machine that transforms a $\xip$-name $\varphi$ of a function in a $\xi$ name can be specified as follows:
			On input $\str a$ it checks if $\str a$ is of the form $\langle n,\sdzero^k\rangle$.
			If so, it returns approximations to the corresponding $\lambda_{n,p}$ by using \Cref{eq:lamp} and the values of the integrals that are encoded in the $\xip$ name.
			If $\str a$ is of the form $0^n$, it returns $0^{\flength \varphi(n+\flength\varphi(n+1)+2)}$.
			If $\str a$ is neither of the above, it returns the empty string.
			The main difficulty in proving that the produced string function is a name of the function $f:= \xip(\varphi)$ is to check that the first condition of (a) in the proof of \Cref{resu:representing banach spaces} is fulfilled.
			By definition of the representation $\xip$ the function $\mu:= \flength \varphi$ is an $\Lp$-modulus of $f$.
			Set $M:=\mu(n+\mu(n+1)+2)$ and define a step function by
			\[ F:= \sum_{i=0}^{2^{M}-1} A_{\mu(n+1)}(f)\big(\frac{2i+1}{2^{M+1}}\big)\cdot\chi_{[\frac i{2^M},\frac{i+1}{2^M}]}. \]
			$A_{\mu(n+1)}(f)$ is an $2^{-n-1}$-approximation to $f$ in $\Lp$-norm by \Cref{resu:approximation} and $F$ is an $2^{-n-1}$-approximation to $f_{\mu(n+1)}$ in supremum norm by \Cref{resu:continuitity of approximations}.
			Thus $F$ is an $2^{-n}$-approximation to $f$ in $\Lp$-norm.
			By \Cref{resu:stepfunctions to haar basis} $F$ can be expanded as a sum of $f_{k,p}$, where each $k$ is smaller than $2^M$.
			This proves that there exists an appropriate approximation.
		\end{proof}

		Two final remarks:
		Firstly the result can be improved to show that $\xi$ and $\xip$ are always polynomial-time equivalent relative to $p$, i.e. if $p$ is provided as an oracle.
		The necessity of having oracle access to approximations for $p$ can be removed by giving up the requirement that the Schauder basis has to be normalized.
		As second remark, \cite{arXiv:1612.06419} also introduces representations of Sobolev spaces, however, these turn out to be more complicated to reconstruct with the methods from this paper.
		The representations from \cite{arXiv:1612.06419} handle Sobolev spaces as space of integrable functions in the sense that the discrete information encoded in the names are still the values of integrals over the function, but the length of a name is required to be an $\Lp$-modulus of the highest derivative of the function instead of an $\Lp$-modulus of the function itself.
		It turns out that if one uses standard Schauder bases of Sobolev spaces in the construction from \Cref{resu:representing banach spaces} (i.e.\ integrals over the Haar basis), one ends up with a representation $\xi$ featuring a similar length of names, but instead of integrals over the function itself, integrals over the highest derivative are provided.
		In terms of polynomial-time translatability this turns out to be strictly less information.}

\section*{Acknowledgment}

	The authors would like to thank Martin Ziegler, Akitoshi Kawamura and Ulrich Kohlenbach for extended discussion of the content of this paper and more.

\bibliography{bib}{}

\end{document}